\documentclass[letterpaper, 11pt,onecolumn,draftclsnofoot]{IEEEtran}
\usepackage{fancyhdr}
\usepackage{amsmath,epsfig}
\usepackage{threeparttable}
\usepackage{epsf,epsfig}
\usepackage{amsmath}
\usepackage{amssymb}
\usepackage{amsfonts}
\usepackage{algorithmic}
\usepackage{algorithm}
\usepackage[noadjust]{cite}
\usepackage{dsfont}
\usepackage{subfigure}

\newtheorem{theorem}{Theorem}

\newtheorem{lemma}{Lemma}
\newtheorem{corollary}{Corollary}

\newtheorem{proposition}{Proposition}

\newtheorem{assumption}{Assumption}

\def\E{\mathbb{E}}

\def\phi{\varphi}

\def\SNR{\mathsf{SNR}}

\def\l{\left}
\def\r{\right}
\def\({\left(}
\def\){\right)}

\def\bff{{\mathbf{f}}}

\def\bh{{\mathbf{h}}}

\def\bs{{\mathbf{s}}}

\def\bv{{\mathbf{v}}}

\def\bx{{\mathbf{x}}}
\def\by{{\mathbf{y}}}

\def\b0{{\mathbf{0}}}

\def\bA{{\mathbf{A}}}
\def\bB{{\mathbf{B}}}
\def\bC{{\mathbf{C}}}

\def\bP{{\mathbf{P}}}
\def\bQ{{\mathbf{Q}}}

\def\bX{{\mathbf{X}}}

\newcommand{\nn}{\nonumber}

\begin{document}

\title{\huge \setlength{\baselineskip}{30pt} Event-Driven Optimal Feedback Control for Multi-Antenna Beamforming}
\author{\large \setlength{\baselineskip}{15pt}Kaibin Huang, Vincent K. N. Lau, and Dongku Kim \thanks{
K. Huang and D. Kim are with the School of Electrical and Electronic Engineering, Yonsei University, 262 Seongsanno, Seodaemun-gu,
Seoul 120-749, Korea. V. K. N. Lau is with the Department of Electronic and Computer Engineering, Hong Kong University of Science and Technology, Clear Water Bay, Hong Kong. Email: huangkb@yonsei.ac.kr, eeknlau@ust.hk, dkkim@yonsei.ac.kr.
}\vspace{-40pt}}

\maketitle
\begin{abstract}\setlength{\baselineskip}{15pt}
Transmit beamforming  is a simple multi-antenna technique for increasing throughput and the transmission range of a wireless communication system. The  required feedback of channel state information (CSI) can potentially result in excessive overhead especially for high mobility or many antennas. This work concerns efficient  feedback for transmit beamforming and establishes a new approach of  controlling  feedback for maximizing net throughput, defined as  throughput minus average feedback cost. The feedback controller using a stationary policy turns CSI feedback on/off according to the system state that comprises the channel state and transmit beamformer.  Assuming channel isotropy and Markovity, the controller's state reduces  to two scalars. This  allows the optimal control policy to be efficiently  computed  using dynamic programming.   Consider the perfect feedback channel free of error, where each feedback instant pays a fixed price.  The corresponding optimal feedback control policy is proved to be of the threshold type. This result holds regardless of whether the controller's state space is discretized or continuous. Under the threshold-type policy, feedback is performed whenever a state variable indicating the accuracy of   transmit CSI is below a threshold, which varies with channel power. The practical finite-rate feedback channel is also considered. The optimal policy for quantized feedback is proved to be also of the threshold type. The effect of CSI quantization is shown to be equivalent to an increment on the feedback price. Moreover, the increment  is upper bounded by the expected logarithm of one minus the quantization error.  Finally,  simulation shows that feedback control  increases net throughput  of the conventional periodic feedback by up to $0.5$ bit/s/Hz without requiring additional bandwidth or antennas.

\end{abstract}


\section{Introduction}\label{Section:Intro}
Transmit beamforming is a popular multi-antenna technique for enhancing the reliability and throughput of a wireless communication link \cite{PaulrajBook}. In many systems, transmit beamforming requires feedback of channel state information (CSI), incurring significant overhead especially for a large number of transmit antennas or fast fading \cite{LovHeaETAL:WhatValuLimiFee:Oct:2004,Love:OverviewLimitFbWirelssComm:2008}. In this paper,  we consider the transmit beamforming system and propose a new  approach of maximizing \emph{net throughput}, defined as throughput minus average feedback cost,  via optimal feedback control. The controller under consideration turns CSI feedback on/off by observing the current system state that consists of the current channel state and transit beamformer.  The optimal stationary control policy is shown to be of the threshold type. As a result, feedback is performed whenever  transmit CSI is sufficiently outdated as measured by an optimal threshold function. Optimal feedback control is observed to substantially  increase net throughput of the transmit beamforming system compared with the conventional periodic feedback \cite{LovHeaETAL:GrasBeamMultMult:Oct:03, LovHeaETAL:WhatValuLimiFee:Oct:2004,MukSabETAL:BeamFiniRateFeed:Oct:03}. 

\subsection{Prior Works}
In multi-antenna systems, adaptive  transmission techniques such as beamforming and precoding typically require periodic feedback of complex vectors or matrices derived from CSI. The potentially large feedback overhead has motivated active research on intelligent algorithms for quantizing feedback CSI, forming a research area called \emph{limited feedback} \cite{Love:OverviewLimitFbWirelssComm:2008}. Different approaches for quantizing CSI  have been proposed, including line packing \cite{LovHeaETAL:GrasBeamMultMult:Oct:03,MukSabETAL:BeamFiniRateFeed:Oct:03}, combined channel parameterization and scalar quantization \cite{Roh:EffFbMIMOParameter:2007}, subspace interpolation \cite{ChoiMondal:InterpPrecodSpaMuxMIMOOFDMLimFb:2006}, and Lloyd's algorithm \cite{Xia:AchieveWelchBound:05, Lau:MIMOBlockFadingFbLinkCapConst:04}. Furthermore, various types of limited feedback systems have been designed, namely  beamforming \cite{LovHeaETAL:GrasBeamMultMult:Oct:03, MukSabETAL:BeamFiniRateFeed:Oct:03}, precoded orthogonal space-time block codes \cite{LoveHeath:LimitedFeedbackPrecodOSTBC:05}, precoded spatial multiplexing \cite{LoveHeath:LimitedFeedbackPrecodSpatialMultiplex:05},
and multiuser downlink \cite{Gesbert:ShiftMIMOParadigm:2007}. The practicality of limited feedback has been recognized by the industry and related techniques have been integrated into latest wireless communication standards such as  IEEE 802.16 \cite{80216} and 3GPP LTE \cite{3GPP-LTE}. 

Besides quantization, CSI feedback can be compressed by exploiting  channel temporal correlation 
\cite{Roh:EffFbMIMOParameter:2007, Banister03, Huang:LimFbBeamformTemporallyCorrChan:2009}.  In \cite{Roh:EffFbMIMOParameter:2007}, each CSI matrix for a  multiple-input-multiple-output (MIMO) channel is parameterized and the parameters are sent back incrementally using the delta modulation. In \cite{Banister03}, the feedback CSI  matrix is compressed to be one bit indicating the channel variation with respect to a reference matrix sent by the transmitter. A lossy feedback compression algorithm is proposed in \cite{Huang:LimFbBeamformTemporallyCorrChan:2009}, which reduces  feedback overhead by omitting  in feedback the infrequent transitions between CSI states. In view of prior works, it remains unknown that how the average feedback cost can be minimized for given  throughput.

 The applications of \emph{opportunism} \cite{Knopp:InfoCapPwrCtrlMu:1995,VisTseETAL:OppoBeamUsinDumb:Jan:02}  to CSI feedback  have resulted in \emph{opportunistic feedback} algorithms for reducing sum feedback overhead in multi-user multi-antenna systems \cite{TangHeath:OppFbDLMuDiv:2005, TangHeath:OppFbMuMIMOLinearRx:2005, Sanayei:OppBeamLimitedFb:2005, Huang:SDMASumFbRate:06, SwannackWorenell:MIMOBroadcastISIT:2006}. The common feature of these algorithms is that CSI feedback is performed only if a channel quality indicator exceeds a fixed threshold. Compared with periodic feedback over dedicated channels 
 (see e.g. \cite{LovHeaETAL:GrasBeamMultMult:Oct:03,MukSabETAL:BeamFiniRateFeed:Oct:03}), opportunistic (aperiodic) feedback is much more efficient in terms of sum feedback overhead and thus is suitable for systems  where users randomly access a common feedback channel. The thresholds for opportunistic feedback can be computed iteratively for maximizing throughput as in \cite{TangHeath:OppFbDLMuDiv:2005, TangHeath:OppFbMuMIMOLinearRx:2005} or derived in closed-form expressions for achieving optimal capacity scaling for asymptotically large numbers of users \cite{Sanayei:MUDiv1BitFb:2005, Sanayei:OppBeamLimitedFb:2005, Huang:SDMASumFbRate:06, SwannackWorenell:MIMOBroadcastISIT:2006}.  For simplicity, the temporal correlation in practical channels  is omitted in existing designs where independent block fading is assumed. Thus the existing opportunistic feedback algorithms are incapable of adapting feedback thresholds to channel dynamics for further feedback reduction.

The common objective of the works mentioned above is to maximize throughput. This performance metric fails to account for feedback cost though feedback competes with data transmission for resources including time, bandwidth and power. Thus net throughput  defined earlier is a more practical metric. In \cite{Love:DuplexDistortLimFbMIMO:2006, Yeung:OptimTradeoffTwoWayLimFbBeam:2009, Xie:OptimalBWAllocDataFeedbackMISO:2006}, net throughput is maximized by  optimizing the  resource allocation to data transmission and feedback. In \cite{Love:DuplexDistortLimFbMIMO:2006}, a two-way beamforming system is considered, where data and CSI flow in  both directions of the link between two multi-antenna transceivers. For this system, bounds on the feedback rate for maximizing net throughput are derived. For a similar system, net throughput is maximized in \cite{Yeung:OptimTradeoffTwoWayLimFbBeam:2009} by optimizing power allocation to training, feedback and data transmission. Net throughput optimization for the beamforming system is also investigated in \cite{Xie:OptimalBWAllocDataFeedbackMISO:2006} in terms of optimal bandwidth allocation to feedback and data transmission. Aligned with the direction of prior works, the current paper addresses net throughput maximization for transmit beamforming from the new perspective of feedback control, which adapts  the mentioned resource allocation to channel dynamics.

\subsection{Contributions and Organization}

In this paper, we consider a single-user transmit beamforming system with multiple transmit and a single receive antennas.  Each feedback instant incurs fixed cost in bit/s/Hz, called \emph{feedback price}. A feedback controller turns the feedback link either on or off such that net throughput is maximized. This work is based on the following assumptions. First, channel realizations form a stationary Markov chain. Second, the channel coefficients are i.i.d. complex Gaussian random variables. This assumption allows the state of the feedback controller to reduce to  two scalars  $g$ and $z$ without compromising the controller's optimality. The parameter $g$ is the channel power and $z$ the squared cosine of the angle between  the transmit beamformer and the channel vector. Large $z$ indicates accurate transmit CSI and vice versa \cite{LovHeaETAL:GrasBeamMultMult:Oct:03,MukSabETAL:BeamFiniRateFeed:Oct:03}. Finally, the distribution of $z$ in the next slot conditioned on a realization $a$ in the current slot is assumed to \emph{stochastically dominate} the counterpart conditioned on  
$b \leq a$ \cite{Bawa75:OptimRulesOrderingProspects}. Essentially,  this assumption implies that $z$ being large in a slot likely  remains large in the next slot. 

The contributions of this paper are summarized as follows. In general, the paper establishes a new approach for controlling feedback for transmit beamforming to maximize net throughput. To efficiently compute the optimal control policy using dynamic programming (DP) \cite{Bertsekas07:DynamicProg},  the state space of the controller, namely the product space of $(g, z)$, is quantized.\footnote{This quantization differs from that for finite-rate feedback considered in Section~\ref{Section:LimFb} and thus has no effect on the quality of feedback CSI.} Consider the perfect feedback channel free of feedback error.   First, given the quantized state space, the feedback control policy for maximizing net throughput is proved to be of the threshold type. Specifically, feedback is performed only if $z$ is below the optimal threshold that depends on $g$. 
Second, the threshold type policy is proved to be optimal for feedback control with the continuous (unquantized) state space. Next, we consider the finite feedback channel that requires feedback CSI quantization. Fourth, the optimality of the threshold-type feedback control policy is proved for quantized feedback.   Feedback CSI quantization reduces  the receive SNR and also varies the dynamics of $z$. Fifth, to gain insight into these two effects, they are treated separately and each of them is shown to decrease net throughput. Finally, we show that the effect of CSI quantization on net throughput  can be interpreted as an increment on the  feedback price. This increment is upper bounded by the expected logarithm of one minus the quantization error. 

Simulation results are also presented for the channel model  specified by i.i.d. Rayleigh fading and Clarke's temporal correlation. Define the \emph{feedback gain} as throughput for free feedback minus that for no feedback. With respect to periodic feedback, optimal controlled feedback is observed to increase the feedback gain by up to $0.5$ bit/s/Hz, equal to $24\%$ of the feedback gain. The increase in net throughput is insensitive to the variation on Doppler frequency and the number of transmit antennas. For both perfect and imperfect feedback channels, the optimal feedback control policies computed numerically are observed to exhibit the threshold structure as predicted analytically. Moreover, the feedback threshold decreases with the increasing feedback price, corresponding to less frequent feedback.  Last,  feedback quantization is observed to  reduce ergodic throughput as well as the feedback threshold, decreasing the feedback frequency. 

The remainder of this paper is organized as follows. The system model is described in Section~\ref{Section:System}. The optimal feedback control policies for the perfect and finite-rate feedback channels are analyzed in Section~\ref{Section:PefectFb} and \ref{Section:LimFb}, respectively. Simulation results are presented in Section~\ref{Section:Sim} followed by concluding remarks in Section~\ref{Section:Conclusion}.

{\bf Notation}: A matrix is represented by a boldface capitalized letter and a vector by   a boldface small letter.  The $(m,n)$th element of  a matrix $\bX$ is represented by  $[\bX]_{m,n}$. For a vector $\bx$, $[\bx]_m$ gives the $m$th element. The superscript $\dagger$  denotes the complex conjugate transpose operation on a matrix or a vector.  Define the operator $(a)^+$ on a scalar $a$ as $(a)^+:= \max(0, a)$. The realization of a stochastic process in the $t$th time slot is specified  by the subscript $t$.

\section{System Model} \label{Section:System}
\begin{figure}[t]
\begin{center}
\includegraphics[width=15cm]{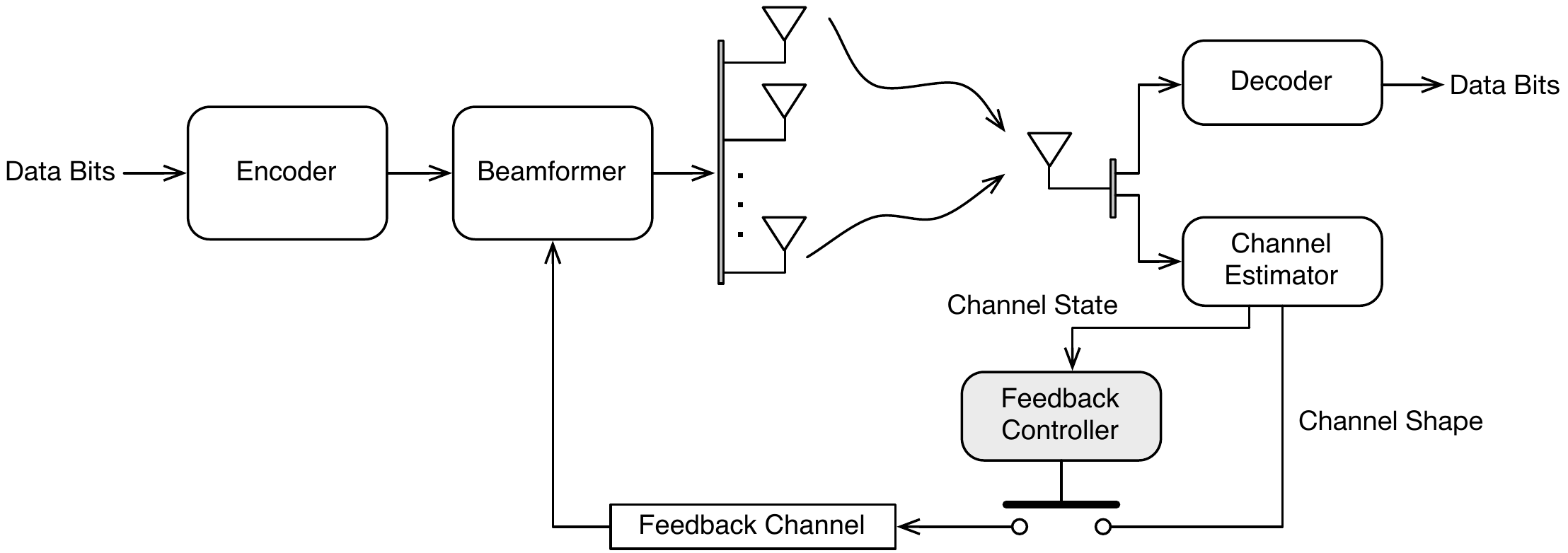}
\caption{Transmit beamforming system with controlled CSI (channel shape) feedback}
\label{Fig:Sys}
\end{center}
\end{figure}

We consider the transmit beamforming system illustrated in Fig.~\ref{Fig:Sys}, where a transmitter with $L$ antennas transmits to a receiver with  a single antenna. The frequency-flat channel is a $L\times 1$ complex vector denoted as $\bh$. To facilitate our designs, $\bh$ is decomposed into the channel power $g := \|\bh\|^2$ and the \emph{channel shape} $\bs := \bh/\|\bh\|$, which are the indicators of the channel quality and direction, respectively. It follows that $\bh = \sqrt{g}\bs$. It is well-known that applying $\bs$ as the beamforming vector, denoted as $\bff$, maximizes the receive signal-to-noise ratio (SNR) \cite{PaulrajBook}. To this end, $\bs$ is estimated by the receiver at the beginning of each time slot of $T_c$ seconds  and communicated to the transmitter via the feedback channel. \footnote{Besides CSI bits, feedback contains an extra bit identifying the feedback instant if the feedback channel is assigned to a single user or multi-bit user identity if multiple users share the feedback channel.} The channel is assumed constant within each slot and thus feedback is performed at most once per slot. \footnote{This requires that $T_c$ is shorter than channel coherence time.}  Depending on the channel state $\bh$ (or $g$ and $\bs$), the feedback controller turns CSI feedback on/off at the beginning of each time slot. Let $\mathcal{U}:=\{0,1\}$ denote the control state space and $\mu \in\mathcal{U}$ the feedback decision, where $1$ and $0$ correspond to the on and off states of the feedback link. Define the controller's state $x := (g, \bs, \bff)$ that contains all system variables affecting net throughput obtained in the sequel. Thus the state space is $\mathcal{X}:=\mathds{R}^+\times\mathds{O}^L\times\mathds{O}^L$ where $\mathds{O}^L$ represents the unit hypersphere embedded in $\mathds{C}^L$ \cite{MukSabETAL:BeamFiniRateFeed:Oct:03}.  We consider a \emph{stationary} feedback control policy $\mathcal{P}:\mathcal{X}\rightarrow\mathcal{U}$ independent of the slot index $t$ \cite{Bertsekas07:DynamicProg}. 
It is assumed that per usage of the feedback channel incurs the feedback cost of $B$ bit. \footnote{The parameter $B$ measures the equivalent number of data bits that can be transmitted reliably using the resources allocated to one-time feedback.  Feedback CSI is delay sensitive and thus cannot be protected by strong error correcting codes as their decoding delay is too long. Therefore, feedback CSI is typically transmitted using larger power and lower-order modulation than those for data transmission. As a result, the communication cost of one CSI bit is higher than that of one data bit ($B>1$). } Moreover, the transmission time of feedback CSI is assumed negligible.\footnote{In practice, feedback CSI is treated as control signals and transmitted in the header that occupies a small fraction of each slot.  This justifies the omission of CSI transmission time. } We consider long  data codewords covering many channel realizations. Given channel ergodicity  and stationary feedback control, the net throughput in bit/s can be written as \cite{Love:DuplexDistortLimFbMIMO:2006}
\begin{equation}
R =\frac{1}{T_s}\E\l[\log_2(1+ Pg|\bs^\dagger \bff|^2)\r] - \frac{B}{T_c}\Pr(\mu = 1) 
\end{equation}
where $P$ is the transmit SNR and $T_s$ the symbol duration. For simplicity, net throughput can be written in bit/s/Hz as  
\begin{equation}\label{Eq:Throughput}
J =\E\l[\log_2(1+ Pg|\bs^\dagger \bff|^2)\r] - \alpha\Pr(\mu = 1) 
\end{equation}
where $\alpha := \frac{BT_s}{T_c}$ is called the feedback price. \footnote{The value of $\alpha$ is large if power allocated to feedback or the number of channel coefficients are large, or the symbol rate is low and vice versa.}

The feedback controller controls the transmit beamformer via feedback and thereby influences the receive SNR.  Define $z := |\bs^\dagger \bff|^2$ that represents the controllable component of the receive SNR $\SNR_r = Pgz$ \cite{LovHeaETAL:GrasBeamMultMult:Oct:03, MukSabETAL:BeamFiniRateFeed:Oct:03}. Consider the perfect feedback channel. The temporal variation of $z$ under feedback control is illustrated in Fig.~\ref{Fig:Control:Z}. Upon CSI feedback, $\bff$ is updated with $\bs$ and the value of $z$ is reset to the maximum of one; if feedback is turned off, $z$ is smaller than one due to that $\bff$ fails to adapt instantaneously to the time varying $\bs$. With $\bff$ fixed, the probability density function (PDF) of $z$ is referred to as the \emph{uncontrolled PDF} and denoted as $\check{f}(z\mid \bff)$. The uncontrolled PDF governs the dynamics of $z$ between two consecutive feedback instants (cf. Fig.~\ref{Fig:Control:Z}). The random variable $z$ is used later as a controller state variable. In addition, besides the mentioned perfect feedback channel, the one with a finite-rate constraint is  also considered in the sequel. 

\begin{figure}[t]
\begin{center}
\includegraphics[width=12cm]{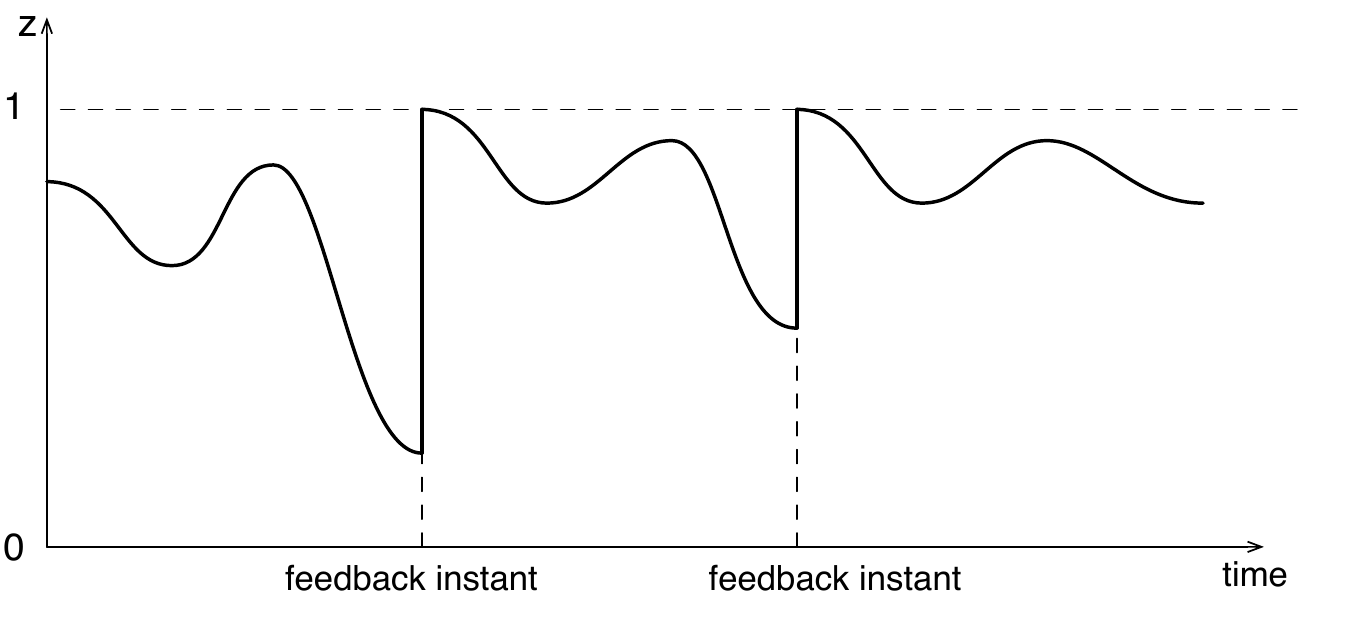}
\caption{Temporal variation of $z$ under feedback control }
\label{Fig:Control:Z}
\end{center}
\end{figure}

We make several assumptions on the channel distribution as described shortly. These assumptions facilitate  computing   the optimal feedback control policy $\mathcal{P}^\star$ using DP \cite{Bertsekas07:DynamicProg} and analyzing  the policy structure. Model the channel as a stochastic sequence denoted as $\bh_0, \bh_1, \bh_2,\dots$, where $\bh_t$ is the channel state  in the $t$th slot. 
\begin{assumption}\label{AS:Markov}
The sequence $\bh_0, \bh_1, \bh_2,\dots$ is a stationary  Markov chain.
\end{assumption}
\noindent In other words, given $\bh_n$, $\bh_{n+1}$ is independent of the past realizations $\bh_{n-1}, \bh_{n-2}, \cdots$.  Markov chains are commonly used for modeling temporally-correlated wireless channels (see e.g., \cite{Wang:FinStatMarkovChan:95, Pimentel:FiniteMarkovModelCorrRicianFading:2004,  Zhang:MarkovModelRayleighFading:1999, Ivanis:AdpMIMOMRCImpCSIMarkov:2007, Kuo:MarkovMIMOCondNumbAntSel:2007}). Markov channel models have been validated both analytically  (see e.g., \cite{Wang:FinStatMarkovChan:95,Sadeghi:CapAnalysisMarkovFlatFading:2005}) and by measurement \cite{Daniels08:ThputDelayMeasurementLimFb:2008}. Next, the controller's state space $\mathcal{X}$  have $(2L+1)$ dimensions. Due to the curse of dimensionality, computing $\mathcal{P}^\star$ is impractical if $L$ is large \cite{Bertsekas07:DynamicProg}. The following assumption overcomes this difficulty, which is commonly made in the literature (see e.g. \cite{TarokhJafETAL:SpacBlocCodeFrom:Jul:99,Tel:CapaMultGausChan:99}). 
\begin{assumption}\label{AS:Gauss}
The channel $\bh$ comprises i.i.d. $\mathcal{CN}(0,1)$ random variables. 
\end{assumption}
\noindent Given this assumption, $g$ follows chi-square distribution with $L$ complex degrees of freedom \cite{Lo:MaxRatioTx:99}; $\bs$ is isotropic. As a result,  the uncontrolled distribution of $z$ is independent of $\bff$ and hence we can write $\check{f}(z\mid \bff)$ as  $\check{f}(z)$. \footnote{Given Assumption~\ref{AS:Gauss}, the uncontrolled distribution  of $z$ conditioned on an arbitrary $\bff$ is  $\Pr(z \geq \tau) = (1-\tau)^{L-1}$ \cite{MukSabETAL:BeamFiniRateFeed:Oct:03,YeungLove:RandomVQBeamf:05}.}
It follows that the state variables $\bs$ and $\bff$ can be combined into $z$. Thus the controller's state and state space  reduce to $x = (g, z)$ and $\mathcal{X} = \mathds{R}^+\times\mathcal{Z}$ with $\mathcal{Z}:=[0,1]$, respectively, thereby overcoming the mentioned curse of dimensionality. We make the following  assumption on the temporal correlation of $z$ that affects the structure of $\mathcal{P}^\star$ (cf. Section~\ref{Section:PefectFb}). 
\begin{assumption}\label{AS:SD}
For $1\geq a\geq  b \geq 0$, the uncontrolled distribution of $z_{t+1}$ conditioned on $z_t = a$ is  stochastically dominant \cite{Bawa75:OptimRulesOrderingProspects} over  that conditioned on $z_t = b$. Mathematically,
\begin{equation}
\int_c^\infty \check{f}(z_{t+1} \mid z_t=a)dz_{t+1} \geq \int_c^\infty \check{f}(z_{t+1} \mid z_t = b)dz_{t+1}\nn
\end{equation}
where $0\leq c\leq 1$.
\end{assumption}
\noindent This assumption essentially states that large  $z_t$ likely leads to large  $z_{t+1}$ and vice versa, which is reasonable given channel temporal correlation. This work requires no assumption on the temporal correlation of $g$. 

Finally, let $f(z_{t+1}\mid z_t, \mu_t)$  and $\tilde{f}(g_{t+1}\mid g_t)$ denote the transition PDF's of $z$ and $g$ respectively. For convenience, the state transition PDF is written as $f_x := \tilde{f}\times f$.

\section{Problem Formulation}\label{Section:ProbForm}
In this section, the problems of optimal feedback control are formulated and  solved in the subsequent sections. Both perfect and finite-rate  feedback channels are considered in the problem formulation. 

The generic average and discounted reward problems are defined as follows \cite{Bertsekas07:DynamicProg}. By abuse of notation, the symbols in the preceding section are reused here. Consider a dynamic system with an infinite number of stages (infinite horizon), a state space $\mathcal{X}$ and a control space $\mathcal{U}$. The system dynamics are specified by the state transition kernel $f_x(x_t\mid x_{t}, \mu_t)$ with $x\in\mathcal{X}$ and $\mu\in\mathcal{U}$. The reward-per-stage is represented by the function $G: \mathcal{X}\times\mathcal{U}\rightarrow\mathds{R}^+$.  The control policy $\mathcal{P}:\mathcal{X}\rightarrow\mathcal{U}$ is optimized for maximizing either the average reward
\begin{equation}\label{Eq:AvAward:PerfFB}
J := \lim_{T\rightarrow\infty}\frac{1}{T} \sum_{t=0}^{T-1}  G(x_t, \mu_t)
\end{equation}
or the discounted reward 
\begin{equation}\label{Eq:DistAward:PerfFB}
J_\beta(x_0) :=  \sum_{t=0}^{\infty}  \beta^t\E[G(x_t, \mu_t)\mid x_0]
\end{equation}
where $0< \beta < 1$ is the discount factor. 
This corresponding optimization problems are called the  infinite-horizon average and discounted  reward problems represented by  $\mathcal{A}(\mathcal{X}, \mathcal{U}, f_x, G)$ and $\mathcal{D}(\mathcal{X}, \mathcal{U}, f_x, G)$, respectively. These problems can be solved iteratively using DP \cite{Bertsekas07:DynamicProg}.

Consider the perfect feedback channel. Net throughput in \eqref{Eq:Throughput} can be written as the average reward in \eqref{Eq:AvAward:PerfFB} with $G$ given by 
\begin{equation}\label{Eq:RewardPerStage}
G(g, z, \mu) := \l\{\begin{aligned}
&\log_2(1+Pg)-\alpha,&& \mu = 1\\
&\log_2(1+Pgz),&& \textrm{otherwise}.
\end{aligned}\right. 
\end{equation}
Hereafter the terms,  average reward and net throughput, are used interchangeably. 
Thus the feedback controller  can be designed by solving  $\mathcal{A}(\mathcal{X}, \mathcal{U}, f_x, G)$ with $f_x$ obtained as follows
\begin{eqnarray}
f_x(x_{t+1}\mid x_t, \mu_t) &=& f_x((g_{t+1}, z_{t+1})\mid (g_t, z_t)) \nn\\
&=& \tilde{f}(g_{t+1}\mid z_{t+1}, g_t, z_t, \mu_t) f(z_{t+1}\mid g_t, z_t, \mu_t) \nn\\
&\overset{(a)}{=}& \tilde{f}(g_{t+1}\mid g_t) f(z_{t+1}\mid z_t, \mu_t) \label{Eq:Kernel}
\end{eqnarray}
where $(a)$ follows from the channel isotropy. From the discussion in Section~\ref{Section:System}
\begin{equation}\label{Eq:PDF:Z}
f(z_{t+1}\mid z_t = c, \mu_t)= \left\{
\begin{aligned}
&\check{f}(z_{t+1} \mid z_t =1),  && \mu_t = 1,\\
&\check{f}(z_{t+1} \mid z_t =c), && \textrm{otherwise}
\end{aligned}\right.
\end{equation}
with $c\in\mathcal{Z}$. The optimal policy $\mathcal{P}^\star$ for solving  $\mathcal{A}(\mathcal{X}, \mathcal{U}, f_x, G)$ is analyzed in Section~\ref{Section:PefectFb}.  \footnote{\label{Foot:MultiObj}The problem of net throughput maximization can be also formulated as the multi-objective optimization problem of   maximizing  ergodic throughput $\E[\log_2(1+Pgz)]$ and minimizing  the feedback rate  $\alpha\Pr(\mu=1)$.  Note that the average feedback rate is proportional to the feedback probability. The multi-objective reward function can be modified from \eqref{Eq:Throughput} by replacing $\alpha$ with $\lambda \alpha$ with $\lambda \geq 0$ being the weight factor. Varying $\lambda$ varies the relative importance of throughput maximization and feedback rate reduction.  Solving the multi-objective optimization problem with varying $\lambda$ gives the maximum throughput as a function of the average feedback rate. However, proving the Pareto optimality \cite{SawaragiBoo:MultiObjOptim:1985} of this function  seems difficult.  }

In practice, CSI feedback is mplemented using a narrow-band (finite-rate) control channel \cite{3GPP-LTE, 80216, LovHeaETAL:WhatValuLimiFee:Oct:2004}. The finite-rate feedback constraint requires quantizing feedback CSI as described shortly.  Let $\hat{\bs}$ denote the $L\times 1$ complex unitary vector resulting from quantizing $\bs$ \cite{LovHeaETAL:GrasBeamMultMult:Oct:03,MukSabETAL:BeamFiniRateFeed:Oct:03}.  We consider a codebook based quantizer where the codebook $\mathcal{F}$ is a set of complex unitary vectors \cite{LovHeaETAL:GrasBeamMultMult:Oct:03,MukSabETAL:BeamFiniRateFeed:Oct:03}. Using $\mathcal{F}$, $\hat{\bs}$ is obtained by quantizing $\bs$ using the maximum SNR criterion, namely that $\hat{\bs} = \max_{\bx \in \mathcal{F}} |\bs^\dagger\bx|^2$. Define $\epsilon := |\hat{\bs}^\dagger\bs|^2$ where $0\leq \epsilon \leq 1$. This scalar quantifies the loss on the receive SNR $(\SNR_\epsilon = Pg\epsilon)$ upon CSI feedback compared with that $(\SNR = Pg)$ for the perfect feedback channel \cite{MukSabETAL:BeamFiniRateFeed:Oct:03,LovHeaETAL:GrasBeamMultMult:Oct:03, YeungLove:RandomVQBeamf:05}. Note that $\epsilon$ is equal to one minus the quantization error defined in \cite{Jindal:MIMOBroadcastFiniteRateFeedback:06,YeungLove:RandomVQBeamf:05}. The distribution of $\epsilon$ depends on the channel distribution  and design of the quantizer codebook   \cite{LovHeaETAL:GrasBeamMultMult:Oct:03,MukSabETAL:BeamFiniRateFeed:Oct:03, YeungLove:RandomVQBeamf:05}. \footnote{\label{Fn:Eps} For example, for the isotropic channel in Assumption~\ref{AS:Gauss} and a randomly generated codebook, the distribution function of $\epsilon$ is $\Pr(\epsilon\geq \delta ) = (1-\epsilon)^{L-1}$ \cite{YeungLove:RandomVQBeamf:05,Jindal:MIMOBroadcastFiniteRateFeedback:06}.}
Again, net throughput maximization is formulated as an average reward problem. This problem differs from   $\mathcal{A}(\mathcal{X}, \mathcal{U}, f_x, G)$ for perfect feedback only in the state transition kernel and award-per-stage. 
The transition PDF of $z$ is obtained as 
\begin{equation}\label{Eq:PDF:Z:Quant}
f_\epsilon(z_{t+1} \mid z_t = c, \mu_t)= \left\{
\begin{aligned}
&\underset{\epsilon}{\E}[\check{f}(z_{t+1} \mid z_t =\epsilon)],  && \mu_t = 1,\\
&\check{f}(z_{t+1} = a\mid z_t =c), && \textrm{otherwise}
\end{aligned}\right.
\end{equation}
where $\underset{\epsilon}{\E}$ denotes the expectation over the distribution of $\epsilon$. 
Thus for the corresponding average reward problem, the state transition kernel is $f_x^\epsilon = \tilde{f}\times f_\epsilon$.  Given the receive SNR $\SNR_\epsilon = Pg\epsilon$ upon CSI feedback, the  reward-per-stage function $G_\epsilon$ is modified from \eqref{Eq:RewardPerStage} as 
\begin{equation}\label{Eq:RewardPerStage:Quant}
G_\epsilon(x, \mu) := \left\{\begin{aligned}
&\underset{\epsilon}{\E}[\log_2(1+Pg\epsilon)]-\alpha, && \mu = 1,\\
&\log_2(1+Pgz), && \textrm{otherwise}.
\end{aligned}\right. 
\end{equation}
In Section~\ref{Section:LimFb}, we consider $\mathcal{A}(\mathcal{X}, \mathcal{U}, f_x^\epsilon, G_\epsilon)$ and analyze the resultant optimal policy. 

\section{Feedback Control Policy: Perfect Feedback Channel}\label{Section:PefectFb}

In this section, the optimal feedback control policy is analyzed for the perfect feedback channel. To compute the policy using DP, the state space of the feedback controller is quantized.   Given the discrete state space, the optimal control policy is proved to be of the threshold type. This policy structure  is shown to also hold for the optimal feedback control with the continuous state space. 
\subsection{State-Space Quantization}\label{Section:QuantAlgo}

The channel state $(g, z)\in\mathcal{X}$ is quantized as $(\hat{g}, \hat{z})\in\hat{\mathcal{X}}$ that is  used as the input of the feedback controller, where $\hat{\mathcal{X}}$  denote the discrete state space defined in the sequel. Feedback control with $\hat{\mathcal{X}}$  allows applying  stochastic optimization theory  to analyzing the optimal control policy in the sequel. \footnote{No comprehensive theory exists for the average cost/reward problem with an infinite or continuous state space \cite{Bertsekas07:DynamicProg}.} The algorithms for quantizing $(g, z)$ are described as follows.  

The space of $g$, namely the nonnegative real line $\mathds{R}^+$, is partitioned into $M$ line segments $[\tilde{g}_0, \tilde{g}_1)$, $[\tilde{g}_1, \tilde{g}_2)$, $\cdots$, $[\tilde{g}_{M-1}, \infty)$, where $\tilde{g}_0=0$ and $0< \tilde{g}_1 < \tilde{g}_2<\cdots <\tilde{g}_{M-1}< \infty$. The values of $\{\tilde{g}_m\}$ are chosen such that $g$ lies in different line segments with equal probabilities. In other words, $\Pr(G\in [\tilde{g}_m, \tilde{g}_{m+1}))=\frac{1}{M}\ \forall \ 0\leq m \leq M-1$ with $\tilde{g}_{M} = \infty$. The above $M$ line segments are represented by a set of  $M$ finite values $\hat{\mathcal{G}} = \{\bar{g}_0,\bar{g}_1, \cdots, \bar{g}_{M-1}\}$ called \emph{grid points} \cite{Bertseka:ConvergeDiscretDynamicProg:75}, which are \emph{arbitrarily} selected from corresponding segments and hence satisfy the constraints $\bar{g}_m\in[\tilde{g}_m, \tilde{g}_{m+1}]\ \forall \ m$.~\footnote{The grid points in the spaces of $g$ and $z$ can be adjusted   to  yield a better approximation of the optimal policy for the continuous state space. However, such an adjustment has  no effect on the analysis in the sequel. }  Similarly, the space of $z$, namely the line segment $\mathcal{Z}=[0, 1]$, is divided into $N$ sub-segments of equal length  $[\tilde{z}_0, \tilde{z}_1), [\tilde{z}_1, \tilde{z}_2), \cdots, [\tilde{z}_{N-1}, \tilde{z}_{N}]$ where $\tilde{z}_0=0$ and $ \tilde{z}_N=1$. \footnote{The line sub-segments are chosen to have equal length rather than equal probability since the distribution of $z$ depends on the optimal feedback control policy and is unknown at this stage.} Define the set $\tilde{\mathcal{Z}} := \{\tilde{z}_n\}_{n=0}^{N}$. Again,  $N$ grid points $\hat{\mathcal{Z}} = \{\bar{z}_0, \bar{z}_1, \cdots, \bar{z}_{N-1}\}$ are arbitrarily  chosen from the $N$ sub-segments mentioned earlier. The discrete state space can be readily written as $\hat{\mathcal{X}} :=\hat{\mathcal{G}}\times\hat{\mathcal{Z}}$. The space $\mathcal{X}$  can be mapped to $\hat{\mathcal{X}}$  using the following quantization functions $\mathcal{Q}_g$ and $\mathcal{Q}_z$
\begin{eqnarray}
\hat{g} = \mathcal{Q}_g(g) &=& \bar{g}_m, \quad g\in[\tilde{g}_m, \tilde{g}_{m+1})\label{Eq:QuantFun:g}\\
\hat{z} = \mathcal{Q}_z(z) &=& \bar{z}_n, \quad z\in[\tilde{z}_n, \tilde{z}_{n+1}). \label{Eq:QuantFun:z}
\end{eqnarray}
Note that the above  quantization algorithms  are used for simplicity and only one of many designs that lead to the same  results as obtained in the following sections. \footnote{Specifically, other quantization algorithms also lead to Theorem~\ref{Theo:Policy:Quant} and Proposition~\ref{Prop:Policy:Cont} if $d_s$ defined in \eqref{Eq:MaxErr} and $\Pr(g\geq \tilde{g}_{M-1})$ converge to zero with $M,N\rightarrow\infty$. See the proofs of  Theorem~\ref{Theo:Policy:Quant} and Proposition~\ref{Prop:Policy:Cont} for details.}

Given Assumption~\ref{AS:Markov}, the sequences $\{\hat{g}_t\}$ and $\{\hat{z}_t\}$ are two  Markov chains with the discrete state spaces $\hat{\mathcal{G}}$ and $\hat{\mathcal{Z}}$, respectively. The transition probabilities of the two Markov chains are decoupled as a result of \eqref{Eq:Kernel}. For $\{\hat{z}_t\}$,   let $P_{m,n}$ denote the probability for transition from the state $m$ to $n$ and $\mu$ the feedback decision corresponding to quantized controller input. Then $P_{m,n}$ can be written as a function of $\mu$
\begin{equation}\label{Eq:TxProb:z}
P_{m,n}(\mu) := \left\{\begin{aligned}
&\int_{\tilde{z}_n}^{\tilde{z}_{n+1}}\check{f}(z_{t+1}=\tau\mid  z_t= 1)d\tau, && \mu = 1\\
&\int_{\tilde{z}_n}^{\tilde{z}_{n+1}}\check{f}(z_{t+1}=\tau\mid  z_t= \bar{z}_m)d\tau, && \textrm{otherwise}  
\end{aligned}\right.
\end{equation}
where $0\leq m, n\leq N-1$. Note that $P_{m,n}(1)$ is independent of $n$. 
Similarly, define the counterpart of  $P_{m,n}$ for $\{\hat{g}_t\}$ as
\begin{equation}
\tilde{P}_{m,n} := \int_{\tilde{g}_n}^{\tilde{g}_{n+1}}\tilde{f}(g_{t+1} = \tau\mid g_t = \bar{g}_m)d\tau\label{Eq:TxProb:g}
\end{equation}
where $0\leq m, n\leq M-1$. Note that $\tilde{P}_{m,n}$ are unaffected by feedback control.  For convenience, define the transition probability matrix $\bP$ with $[\bP]_{m,n} := P_{m,n}$ and similarly $\tilde{\bP}$ with $[\tilde{\bP}]_{m,n} := \tilde{P}_{m,n}$. Due to feedback control, the stationary probabilities of $\hat{z}$ depend on those of $\hat{g}$. Thus we define the joint stationary probability $\pi_{m,n} := \Pr(\hat{g} = \bar{g}_m, \hat{z}=\bar{z}_n)$ where $0\leq m\leq M-1, 0\leq n\leq N-1$. With the discrete state space, the state transition kernel is denoted  as $\bP_x := \tilde{\bP}\times\bP$.

The average reward problems $\mathcal{A}(\hat{\mathcal{X}}, \mathcal{U}, \bP_x, G)$ and $\mathcal{A}(\mathcal{X}, \mathcal{U}, f_x, G)$ are considered in Section~\ref{Section:PefectFb:QSpace} and \ref{Section:PefectFb:CSpace}, respectively. Let $\hat{J}^\star$ denote the maximum average reward for the discrete state space.  Then $\hat{J}^\star$ is an approximation of $J^\star$  for the continuous state space. They converge as the quantization resolution increases: $M\rightarrow\infty, N\rightarrow\infty$  (cf. Section~\ref{Section:PefectFb:CSpace}). 

\subsection{Policy for Discrete State Space}\label{Section:PefectFb:QSpace}
This section focuses on $\mathcal{A}(\hat{\mathcal{X}}, \mathcal{U}, \bP_x, G)$. The resultant optimal policy $\hat{\mathcal{P}}^\star$  is shown to be of the threshold type. In addition, the computation of  $\hat{\mathcal{P}}^\star$ is discussed. 

Rather than obtaining $\hat{\mathcal{P}}^\star$ directly, the policy $\hat{\mathcal{P}}^\star_\beta$ is derived by solving $\mathcal{D}(\hat{\mathcal{X}}, \mathcal{U}, \bP_x, G)$. Then the desired $\hat{\mathcal{P}}^\star$ follows from $\hat{\mathcal{P}}^\star_\beta$ by allowing $\beta\rightarrow 1$. The policy $\hat{\mathcal{P}}^\star_\beta$ can be found using DP \cite{Bertsekas07:DynamicProg}. To this end, define the DP operator $\mathsf{F}$ on a given function $q:\hat{\mathcal{X}}\rightarrow\mathds{R}^+$ as 
\begin{equation}\label{Eq:DPOperator}\begin{aligned}
(\mathsf{F}q)(\bar{g}_m, \bar{z}_n)=\max_{\mu\in\{0,1\}}\l[G(\bar{g}_m, \bar{z}_n, \mu)+
\beta\sum\nolimits_{k,\ell} q(k, \ell)\tilde{P}_{k, m} P_{\ell, n}(\mu)\right]\end{aligned}
\end{equation}
where $0\leq m \leq M-1$ and $0\leq n \leq N-1$. The maximum discounted reward $\hat{J}_\beta^\star$ satisfies Bellman's equation $\hat{J}_\beta^\star = \mathsf{F}\hat{J}_\beta^\star$ \cite{Bertsekas07:DynamicProg}. For convenience, represent $\hat{J}_\beta((\bar{g}_m, \bar{z}_n))$ by $\hat{J}_\beta(m, n)$.  

We refer to a $N\times N$ \emph{stochastic matrix} $\bA$ as being \emph{montone} if $\bA$ satisfies~\footnote{In this paper, a stochastic matrix refers to the right stochastic matrix that comprises nonnegative elements and the sum of each column is equal to one \cite{GallagerBook:StochasticProcs:95}. }
\begin{equation}
\sum_{m=m_0}^{N-1}[\bA]_{m,n_1} \geq \sum_{m=m_0}^{N-1}[\bA]_{m,n_2}\quad \textrm{if}\ 0\leq n_2\leq n_2\leq N-1
\end{equation}
where $0\leq m_0\leq N-1$. Thus $\bP$ is monotone following Assumption~\ref{AS:SD} and \eqref{Eq:TxProb:z}. Moreover, define a \emph{monotone vector} of real numbers as one whose elements are in the \emph{ascending} order. The following lemma is useful for the analysis in this paper. 
\begin{lemma}\label{Lem:MonoProd}\
\begin{enumerate}
\item Consider a real vector $\bv$ and a stochastic matrix $\bA$ that are both monotone and have the same height. Then $\bA^\dagger \bv$ is a monotone vector;
\item Consider a matrix $\bB$ of nonnegative elements and a matrix  $\bC$ with monotone rows. Then the rows of $\bB\bC$ are also monotone. 
\end{enumerate}
\end{lemma} 
\begin{proof}See Appendix~\ref{App:MonoProd}. 
\end{proof}

The following lemma is essential for obtaining  the main result of this section. The proof of Lemma~\ref{Lemma:Monotone} is based on \emph{value iteration} \cite{Bertsekas07:DynamicProg}. Using this method, for an arbitrary function $q:\hat{\mathcal{X}}\rightarrow\mathds{R}^+$, the maximum discounted reward  $\hat{J}_\beta^\star$ can be computed iteratively as 
\begin{equation}\label{Eq:ValIt}
\hat{J}_\beta^\star(m,n) = \lim_{k\rightarrow\infty}(\mathsf{F}^k q)(m,n).
\end{equation} 
\begin{lemma}\label{Lemma:Monotone} $\hat{J}_\beta^\star$ has the following properties:
\begin{enumerate}
\item Given $\hat{g}$, $\hat{J}^\star_\beta(\hat{g},\hat{z})$ monotonically increases with $\hat{z}$;
\item Define $w(m,n, \mu) := \sum_{k,\ell} \hat{J}^\star_\beta(k,\ell)\tilde{P}_{m,k}P_{n,\ell}(\mu)$. Given $m$ and $\mu$, $w(m,n, \mu)$ monotonically increases with $n$;
\item Given $m$, $w(m,n, 1) \geq w(m,n, 0)\ \forall \ n$.
\end{enumerate}
\end{lemma}
\begin{proof}
See Appendix~\ref{App:Monotone}.
\end{proof}

Using the above lemma, the main result of this section is obtained as shown in the following theorem. 
\begin{theorem} \label{Theo:Policy:Quant}The optimal policy $\hat{\mathcal{P}}^\star$ is of the threshold type. Specifically, there exists a function $\hat{y}: \hat{\mathcal{G}}\rightarrow \mathcal{Z}$ such that 
\begin{equation}
\hat{\mathcal{P}}^\star: \mu = \left\{\begin{aligned}
&0, && \hat{z} \geq \hat{y}(\hat{g})\\
&1,&& \textrm{otherwise.}
\end{aligned}\right.\label{Eq:Policy:Quant}
\end{equation}
The function $\hat{y}(\cdot)$ is bounded as 
 \begin{equation}\label{Eq:Threshold:Bounds:Dist}
\left(\frac{2^{-\alpha}(1+P\hat{g})-1}{P\hat{g}}\r)^+\leq \hat{y}(\hat{g}) \leq 1.
\end{equation} 
\end{theorem}
\begin{proof}
See Appendix~\ref{App:Policy:Quant}. 
\end{proof} 
Several  remarks are in order. 
\begin{enumerate}
\item Why the optimal policy has the threshold structure is explained as follows. Small $\hat{z}$ corresponds to outdated transmit CSI and vice versa. As a result, CSI feedback for small $\hat{z}$ yields significant reward-per-stage but that for  large $\hat{z}$ may result in negative reward-per-stage due to the feedback cost. Therefore the optimal feedback policy should enable feedback only in the regime of small $\hat{z}$, resulting in the threshold-type policy. In addition, the feedback threshold on $\hat{z}$ depends on $\hat{g}$ since the reward-per-stage is a function of $\hat{g}$.  

\item One would expect that larger $\hat{g}$ makes feedback more desirable because the resultant reward is also larger (cf. \eqref{Eq:RewardPerStage}). In other words, the threshold function $\hat{y}$ should monotonically increase with $\hat{g}$. This is observed in simulation. However, proving this property requires making an assumption on the temporal correlation of $g$, which, however, is unnecessary for this work. 

\item The threshold lower bound in  \eqref{Eq:Threshold:Bounds:Dist} corresponds to the feedback control policy that enables feedback whenever it gives larger reward-per-stage than  no feedback. However,  this policy is suboptimal because feedback may lead to extra reward in subsequent slots despite providing a smaller reward-per-stage in the current slot than no feedback. Therefore the optimal policy should support more frequent feedback than the above suboptimal one. This is the reason that the optimal threshold function is lower bounded  as shown in \eqref{Eq:Threshold:Bounds:Dist}.

\item For a high transmit SNR ($P\rightarrow\infty$), the award-per-stage in \eqref{Eq:RewardPerStage} can be approximated as 
\begin{eqnarray}
G(\hat{g}, \hat{z}, \mu) &=& \left\{
\begin{aligned}
&\log_2(1+P\hat{g}) - \alpha,&& \mu = 1,\\
&\log_2(1/\hat{z}+P\hat{g}) + \log_2\hat{z},&& \textrm{otherwise}
\end{aligned}\right.\\
 &\approx& \left\{
\begin{aligned}
&\log_2(P\hat{g}) - \alpha,&& \mu = 1,\\
&\log_2(P\hat{g}) + \log_2\hat{z},&& \textrm{otherwise}.
\end{aligned}\right. \label{Eq:G:HiSNR}
\end{eqnarray}
Define $\Delta G(\hat{g}, \hat{z}) := G(\hat{g}, \hat{z}, 0)-G(\hat{g}, \hat{z}, 1)$. From \eqref{Eq:G:HiSNR}, $\Delta G(\hat{g}, \hat{z})\approx  \log_2\hat{z} + \alpha$ for $P\rightarrow \infty$. The optimal policy $\hat{\mathcal{P}}^\star$ essentially depends only on $\Delta G(\hat{g}, \hat{z}) $  and the dynamics of $\hat{z}$. Both factors are independent of $\hat{g}$  for $P\rightarrow \infty$. Consequently, the optimal threshold on $\hat{z}$ is insensitive to the variation on $\hat{g}$ for high SNR's. This is confirmed by simulation as discussed in Section~\ref{Section:Sim}.  

\end{enumerate}

For the extreme cases of zero and infinite feedback prices, the feedback threshold is specified in the following corollary. 
\begin{corollary}\label{Cor:Policy:ExtremeFb}
The feedback threshold $\hat{y}$ is fixed at $\hat{y} = 0$ for $\alpha = \infty$ and $\hat{y} = 1$ for $\alpha = 0$.
\end{corollary}
\begin{proof} See Appendix~\ref{App:Policy:ExtremeFb}
\end{proof}
Note that $\hat{y}=0$ and $\hat{y}=1$ correspond to no feedback and feedback in every time slot, respectively. The above results agree with the intuition that feedback is undesirable when the feedback price is too high but feedback should be performed persistently if it is free.

Given the threshold function $\hat{y}$ defining $\hat{\mathcal{P}}^\star$ (cf. Theorem~\ref{Theo:Policy:Quant}),  the maximum award can be obtained as
\begin{equation}\label{Eq:AvReward:b}
\hat{J}^\star(\hat{y}) = \sum_{m=0}^{M-1}\sum_{n=0}^{N-1}G(\bar{g}_m, \bar{z}_n, \check{\mu}_{m,n})\pi_{m,n} 
\end{equation}
where 
\begin{equation}
\check{\mu}_{m,n} = \l\{\begin{aligned}
&0,&& \bar{z}_n\geq \hat{y}(\bar{g}_m),\\
&1,&& \textrm{otherwise}
\end{aligned}
\r.\label{Eq:Decision:mn}
\end{equation}
and the  stationary probabilities $\{\pi_{m,n}\}$ are obtained by solving the following linear equations \cite{GallagerBook:StochasticProcs:95}
\begin{equation}\label{Eq:Eqs}
\pi_{m,n} = \sum_{k,\ell}\tilde{P}_{m,k}P_{n,\ell}(\check{\mu}_{m,n})\pi_{k,\ell}\quad \textrm{and} \quad \sum_{m,n}\pi_{m,n} = 1.
\end{equation}

Finally, we discuss the computation of $\hat{\mathcal{P}}^\star$. Let  $\hat{\by}$ denote the $N\times 1$ threshold vector with $[\hat{\by}]_n = \hat{y}(\bar{g}_n)$.  Then  $\by$  determining $\hat{\mathcal{P}}^\star$ can be computed either by an exhaustive search or policy iteration \cite{Bertsekas07:DynamicProg}. Using Theorem\ref{Theo:Policy:Quant} and \eqref{Eq:AvReward:b}, the brute force approach is specified by 
$\hat{\by} = \max_{\bx\in\mathcal{V}}J^\star(\bx)$
where $\mathcal{V} = \left\{\bx \in \tilde{\mathcal{Z}}^M\mid [\bx]_m \geq \l(\frac{2^{-\alpha}(1+P\bar{g}_m)-1}{P\bar{g}_m}\r)^+\right\}$. 
For this approach, the threshold structure of $\hat{\mathcal{P}}^\star$  is exploited  to reduce the complexity of the exhaustive search from $O(2^{MN})$ to $O(|\mathcal{V}|) = O(N^M)$.  However, this complexity is still too high if $M$ and $N$ are large. For this case, a more practical approach for computing $\by$ is policy iteration  \cite{Bertsekas07:DynamicProg}. Each iteration comprises two steps, namely \emph{policy evaluation} and \emph{policy improvement}. The policy evaluation in the $i$th iteration is to  evaluate a given policy $\hat{\mathcal{P}}^{(i)}$ by computing the average reward and a set of parameters $\{A_{m,n}\}$ called \emph{differential rewards} as follows \cite{Bertsekas07:DynamicProg}
\begin{eqnarray}
J^{(i+1)} + A^{(i+1)}_{m,n} &=& G(\bar{g}_m, \bar{z}_n, \mu_{m,n}) + \sum_{k,\ell}A^{(i)}_{k,\ell}\tilde{P}_{k,m}P_{\ell,n}(\mu_{m,n}),\quad \forall \ m,n\\
A^{(k)}_{M-1, N-1} &=& 0
\end{eqnarray}
where $\mu_{m,n}$ denotes the decision for the state $(\hat{g}, \hat{z}) = (\bar{g}_m, \bar{z}_n)$. 
The subsequent \emph{policy improvement} is specified by 
\begin{equation}
\mu^{(i+1)}_{m,n} = \arg\max_{x\in\{0,1\}} \l[G(\bar{g}_m, \bar{z}_n, x) + \sum_{k,\ell}A^{(i)}_{k,\ell}\tilde{P}_{k,m}P_{\ell,n}(x)\r],\quad \forall \ m,n.
\end{equation}
The policy iteration terminates if $\mu^{(i+1)}_{m,n} = \mu^{(i)}_{m,n}\ \forall \ m,n$. For the simulation in Section~\ref{Section:Sim}, the policy iteration converges typically within several iterations.

\subsection{Policy for Continuous State Space}\label{Section:PefectFb:CSpace}
In this section, we consider the case where $(g, z)$ is directly used as the controller input and design the controller by solving $\mathcal{A}(\mathcal{X}, \mathcal{U}, f_x, G)$. The resultant optimal feedback control policy $\mathcal{P}^\star$ is proved to be of the threshold type.   Specifically, we show that the threshold structure of $\hat{\mathcal{P}}^\star$ as given in Theorem~\ref{Theo:Policy:Quant} holds in the limit  of high quantization resolution ($M\rightarrow\infty$ and $N\rightarrow\infty$).

The proof of this result uses those in \cite{Bertseka:ConvergeDiscretDynamicProg:75}, which addresses the validity of approximately solving  a discounted-reward (or discounted-cost) problem with  a continuous state space by quantizing the space and using DP. A key result in \cite{Bertseka:ConvergeDiscretDynamicProg:75} states that the approximate solution converges to the continuous-space counterpart as the space-quantization error reduces to zero. This  requires that the reward-per-stage and the state transition kernel are Lipschitz continuous. To state this result mathematically, some notation is introduced. Consider an infinite-horizon discounted reward problem with a compact state space $\mathcal{X}'$ and a finite control space $\mathcal{U}'$. Let $x'\in \mathcal{X}'$ denote the state with a transition PDF  $\acute{f}(x'_{t+1}\mid x'_t, \mu'_t)$ for $\mu'_t\in\mathcal{U}'$. Given a  set of grid points $\hat{\mathcal{X}}'$ in $\mathcal{X}'$, $\hat{x}'\in\hat{\mathcal{X}}'$ results from quantizing $x'$, namely that $\hat{x}' = \mathcal{Q}(x'):=\min_{a\in\hat{\mathcal{X}}'}\|x'-a\|^2$. Let the matching state transition kernel be represented by $\bP'_x$.  Define the maximum quantization error as $d_s := \max_{x\in \mathcal{X}'}\max_{\hat{x}'\in\hat{\mathcal{X}}'}\|x' - \hat{x}'\|$.  Let $E_\beta$ and $\hat{E}_\beta$ denote the discounted rewards obtained by solving $\mathcal{D}(\mathcal{X}', \mathcal{U}', \acute{f}, \acute{G})$ and $\mathcal{D}(\hat{\mathcal{X}'}, \mathcal{U}', \bP'_x, \acute{G})$, respectively, where $\acute{G}$ is a reward-per-stage function. A key result in \cite{Bertseka:ConvergeDiscretDynamicProg:75} is stated in the following lemma. 
\begin{lemma}[\cite{Bertseka:ConvergeDiscretDynamicProg:75}]\label{Lem:Converge}  
Assume the reward-per-stage function $\acute{G}$ and $\acute{f}$ satisfy the following Lipschitz conditions
\begin{eqnarray}
\|\acute{G}(a, \mu') - \acute{G}(b, \mu')\| &\leq& V \| a - b\| \nn\\
\|\acute{f}(x_{t+1}\mid x_t=a, \mu_t) - \acute{f}(x_{t+1}\mid x'_t=b, \mu_t)\| &\leq& W \|a-b\|\nn
\end{eqnarray}
where $a, b\in\mathcal{X}$, and $V$ and $W$ are positive constants. Then 
\begin{equation}\label{Eq:Converge}
\lim_{d_s\rightarrow 0}\sup_{x'\in\mathcal{X}'}\l|  E_\beta(x') - \hat{E}_\beta(\mathcal{Q}(x'))\r|=0. 
\end{equation}
\end{lemma}

Lemma~\ref{Lem:Converge} cannot be directly applied to extending the threshold structure of $\hat{\mathcal{P}}^\star$ in Theorem~\ref{Theo:Policy:Quant}  to the continuous-space  counterpart $\mathcal{P}^\star$. The reason is that the continuous state space $\mathcal{X}$ is unbounded and thus not compact. As a result, it is not guaranteed that  $d_s\rightarrow 0$ for $M,N\rightarrow\infty$, which, however, is required for the convergence in \eqref{Eq:Converge}.

The main result of this section is given in the following proposition. To overcome the mentioned difficulty on directly applying Lemma~\ref{Lem:Converge}, the proof of Proposition~\ref{Prop:Policy:Cont} uses a dummy stochastic optimization problem with a bounded and continuous state space. The average reward of this problem is shown to converge to that of the target problem with a unbounded state space as the quantization resolution increases, proving the desired result. 
\begin{proposition} \label{Prop:Policy:Cont}If the transition PDF's $\tilde{f}(g_{t+1}\mid g_t)$ and $\check{f}(z_{t+1}\mid z_t)$ are Lipschitz continuous, the optimal policy $\mathcal{P}^\star$ is of the threshold type. Specifically, there exists a function $y:\mathcal{G}\rightarrow\mathcal{Z}$ such that 
\begin{equation}
\mathcal{P}^\star: \mu = \left\{\begin{aligned}
&0, && z \geq y(g)\\
&1,&& \textrm{otherwise.}
\end{aligned}\right.\label{Eq:Policy:Cont}
\end{equation}
Moreover, $y(\cdot)$ is bounded as 
 \begin{equation}\label{Eq:Threshold:Bounds}
\l(\frac{2^{-\alpha}(1+Pg)-1}{Pg}\r)^+\leq y(g) \leq 1.
\end{equation} 
\end{proposition}
\begin{proof}
See Appendix~\ref{App:Policy:Cont}. 
\end{proof}
We offer the following remarks. 
\begin{enumerate}
\item The feedback probability $\Pr(\mu = 1)$ is strictly larger than zero based on the following argument. For an arbitrary value of $z$, there exist $x>0$ such that the rate function $\log_2(1+Pg)-\alpha>\log_2(1+Pgz)  \ \forall \ g\geq x$, corresponding to $\mu=1$ (cf. \eqref{Eq:RewardPerStage} and Lemma~\ref{Lemma:Monotone}). Since $g$ follows the chi-square distribution, $\Pr(g> x)  >0$ and thus $\Pr(\mu = 1) > 0$. This justifies the above claim. 

\item The continuous-space policy $\mathcal{P}^\star$ cannot be directly computed  using DP but can be approximated by interpolating the discrete-space counterpart $\hat{\mathcal{P}}^\star$ in Theorem~\ref{Theo:Policy:Quant} \cite{Bertsekas07:DynamicProg}. The approximation accuracy improves   with increasing  quantization resolution specified by $M$ and $N$ at the cost of rapidly growing computation complexity. 
 
\end{enumerate}

\section{Feedback Control Policy: Finite-Rate Feedback Channel}\label{Section:LimFb}

A perfect feedback channel is assumed for the analysis in the preceding section. In this section, we consider a finite-rate feedback channel. The optimal feedback control policy is shown to remain as the threshold type. The maximum average reward for finite-rate feedback is shown to be equal to that for perfect feedback at an increased feedback price. Feedback control considered in this section has the discrete state space $\hat{\mathcal{X}}$ as defined in Section~\ref{Section:QuantAlgo}. The results in this section can be extended straightforwardly  to feedback control with the continuous state space $\mathcal{X}$ following the approach in Section~\ref{Section:PefectFb:CSpace}. The details are omitted for brevity.

Consider the average reward problem $\mathcal{A}(\hat{\mathcal{X}}, \mathcal{U}, \tilde{\bP}\times \bP_\epsilon, G_\epsilon)$ that approximates  $\mathcal{A}(\mathcal{X}, \mathcal{U}, f_x^\epsilon, G_\epsilon)$ formulated in Section~\ref{Section:ProbForm}, where $G_\epsilon$ is in 
\eqref{Eq:RewardPerStage:Quant} and 
\begin{equation}\label{Eq:TransProb:Eps}
[\bP_\epsilon(1)]_m = \int_{\tilde{z}_m}^{\tilde{z}_{m+1}} \underset{\epsilon}{\E}[f(\hat{z}_{t+1}=\tau\mid \hat{z}_t = \epsilon)] d\tau
\end{equation}
and $\bP_\epsilon(0) = \bP(0)$. As specified in the following lemma, $\bP_\epsilon$ and $G_\epsilon$ are observed from \eqref{Eq:TransProb:Eps} to have the same properties as their counterparts for the case of perfect feedback considered in Section~\ref{Section:PefectFb}. 

\begin{lemma}\label{Lem:Mono:QFb}\
\begin{enumerate}
\item For $\hat{x} = (\hat{g}, \hat{z})$, $G_\epsilon(\hat{x}, \mu)$ monotonically increases with $\hat{z}$; $G_\epsilon(\hat{x}, 1)$ is independent of $\hat{z}$. 
\item $\bP_\epsilon$ is monotone and has identical columns. 
\end{enumerate}
\end{lemma}
Using Lemma~\ref{Lem:Mono:QFb}, we have the following corollary of Theorem~\ref{Theo:Policy:Quant}.
\begin{corollary}\label{Cor:Policy:QuantFb}
For quantized feedback, the optimal feedback control policy $\hat{\mathcal{P}}^\star_\epsilon$ resulting from solving $\mathcal{A}(\hat{\mathcal{X}}, \mathcal{U}, \tilde{\bP}\times \bP_\epsilon, G_\epsilon)$ is of the same threshold type as specified in Theorem~\ref{Theo:Policy:Quant}.
\end{corollary}
Given feedback inaccuracy due to quantization, feedback may not be always desirable even if it is free ($\alpha=0$). Thus the first claim in Corollary~\ref{Cor:Policy:ExtremeFb} does not hold for quantized feedback as confirmed by simulation (cf. Fig.~\ref{Fig:ThresholdVsPrice}).

It can be observed from \eqref{Eq:RewardPerStage:Quant} and \eqref{Eq:TransProb:Eps} that quantized feedback affects both the award-per-stage and the dynamics of $\hat{z}$. The joint effects on the maximum average reward cannot be characterized using simple expressions. To provide insight into these effects, they are analyzed separately. For this purpose, define the function $\hat{J}^\star(A,  \bB)$ that gives the maximum average reward of $\mathcal{A}(\hat{\mathcal{X}}, \mathcal{U}, \tilde{\bP}\times \bB, A)$.    Then the effects of feedback quantization are specified in the following proposition. 
\begin{proposition}\label{Prop:AvAward:LimFb}
 The function $\hat{J}^\star(\cdot, \cdot)$ satisfies the following inequalities:  
 \begin{enumerate}
\item $\hat{J}^\star(G, \bP) \geq \hat{J}^\star(G, \bP_\epsilon) \geq \hat{J}^\star(G_\epsilon, \bP_\epsilon)$
\item $\hat{J}^\star(G, \bP) \geq \hat{J}^\star(G_\epsilon, \bP) \geq \hat{J}^\star(G_\epsilon, \bP_\epsilon)$
\item $\hat{J}^\star(G_\epsilon, \bP, \alpha) \geq \hat{J}^\star(G, \bP, \alpha-\E[\log_2\epsilon])$
\item $\hat{J}^\star(G_\epsilon, \bP_\epsilon, \alpha) \geq \hat{J}^\star(G, \bP_\epsilon, \alpha-\E[\log_2\epsilon])$. 
\end{enumerate}
\end{proposition}
\begin{proof}
See Appendix~\ref{App:AvAward:LimFb}. 
\end{proof}
\noindent The inequalities in $1)$ and $2)$ state that both effects of finite-rate feedback on the award-per-stage and the dynamics of $\hat{z}$ reduce the maximum average reward with respect perfect feedback. As implied by the inequalities in $3)$ and $4)$, the reward reduction due to feedback quantization  is equivalent to that caused by the increase on the feedback price by at most the amount of $\E[\log_2\epsilon]$. For the specific  distribution of $\epsilon$ in Footnote~\ref{Fn:Eps} and $|\mathcal{F}|\gg 1$, this quantity can be approximated as \cite{Jindal:MIMOBroadcastFiniteRateFeedback:06,YeungLove:RandomVQBeamf:05}
\begin{equation}
\E[\log_2\epsilon] \approx \log_2e\times (1- \E[\epsilon])< \log_2e\times |\mathcal{F}|^{-\frac{1}{M-1}}  . 
\end{equation}
Thus for $|\mathcal{F}| \rightarrow \infty$, $\E[\log_2\epsilon]\rightarrow 0$ and the \emph{equalities} in $3)$ and $4)$ of Proposition~\ref{Prop:AvAward:LimFb} hold.

\section{Simulation Results}\label{Section:Sim}
In this section, additional insight into optimal feedback control are obtained from simulations results. In the simulation, the channel model follows Assumption~\ref{AS:Gauss}. Their temporal correlation is specified by Clarke's function \cite{Clarke74}. The state space for feedback control is quantized as discussed in Section~\ref{Section:QuantAlgo} with $M=N=16$.  The transmit SNR is $20$ dB.

\begin{figure}
\centering\includegraphics[width=11cm]{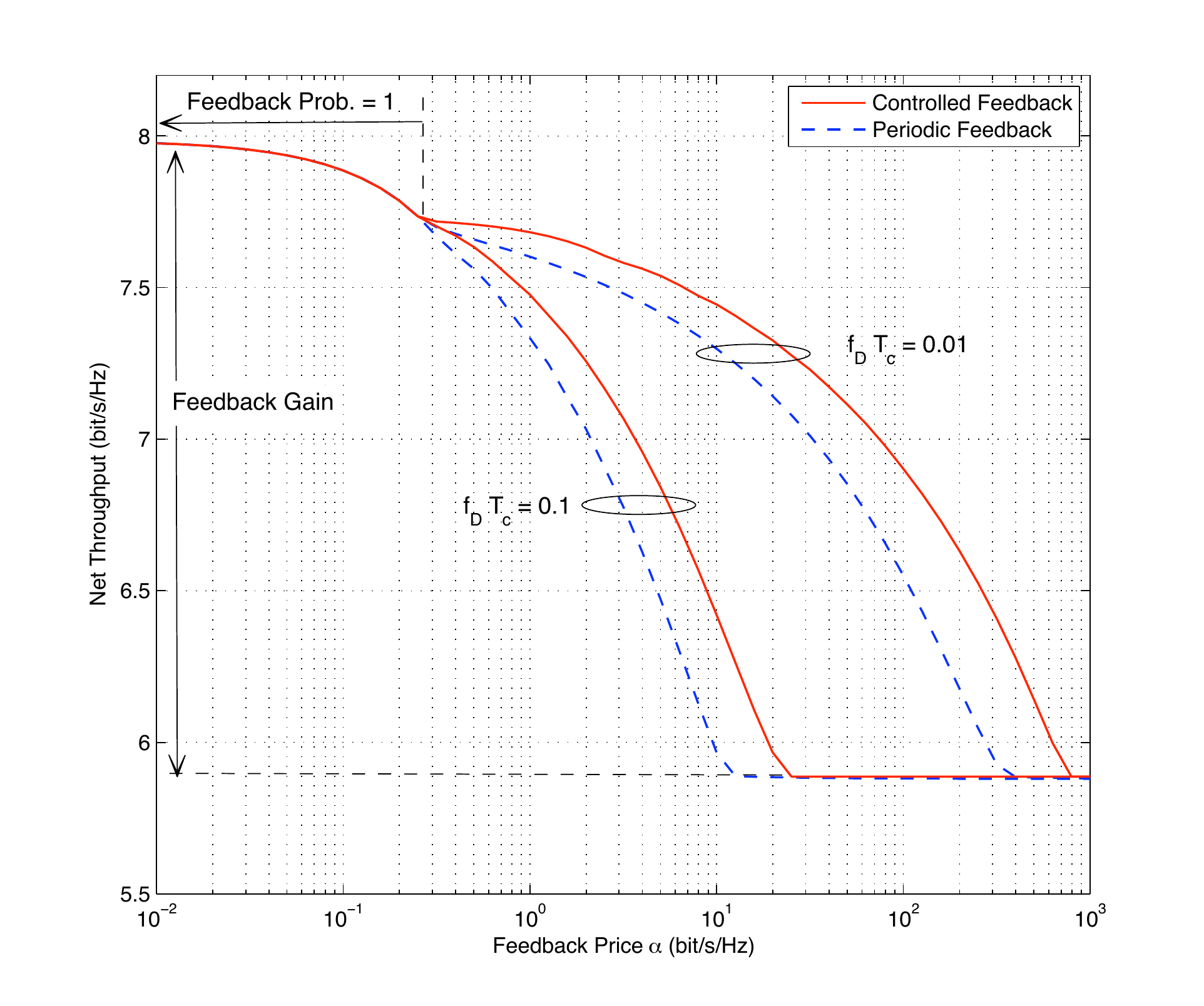}
\caption{Net throughput versus feedback price for both the optimally controlled and periodic feedback over the perfect feedback channel. The Doppler frequency is $f_D = \{0.1, 0.01\}/T_c$ and the number of transmit antenna $L=3$. }
\label{Fig:ThputVarDop}
\end{figure}

In Fig.~\ref{Fig:ThputVarDop}, the curves of net throughput versus feedback price are plotted for both the optimally controlled  feedback  and the conventional periodic feedback. The net throughput for 
controlled and periodic feedback are maximized by  value iteration \cite{Bertsekas07:DynamicProg} and a numerical search over different feedback intervals, respectively. 
The Doppler frequency is $f_D = \{0.1, 0.01\}/T_c$ and the number of transmit antenna $L=3$. As observed from Fig.~\ref{Fig:ThputVarDop}, the throughput for all cases decreases with the increasing feedback price. For high feedback prices, the curves flatten  with net throughput fixed at $5.9$ bit/s/Hz, corresponding to no feedback. 
Subtracting this value from net throughput gives the feedback gain as indicated in Fig.~\ref{Fig:ThputVarDop}.  Controlled feedback is observed to increase net throughput of periodic feedback by up to $0.5$ bit/s/Hz or $24\%$ of the feedback gain of about $2.1$  bit/s/Hz. The increment in net throughput  is insensitive to the change on Doppler frequency. Finally,  for small feedback prices ($\alpha \leq 0.15$), both feedback algorithms perform feedback in  every slot and thus all curves in Fig.~\ref{Fig:ThputVarDop} overlap in this range.

\begin{figure}
\centering\includegraphics[width=11cm]{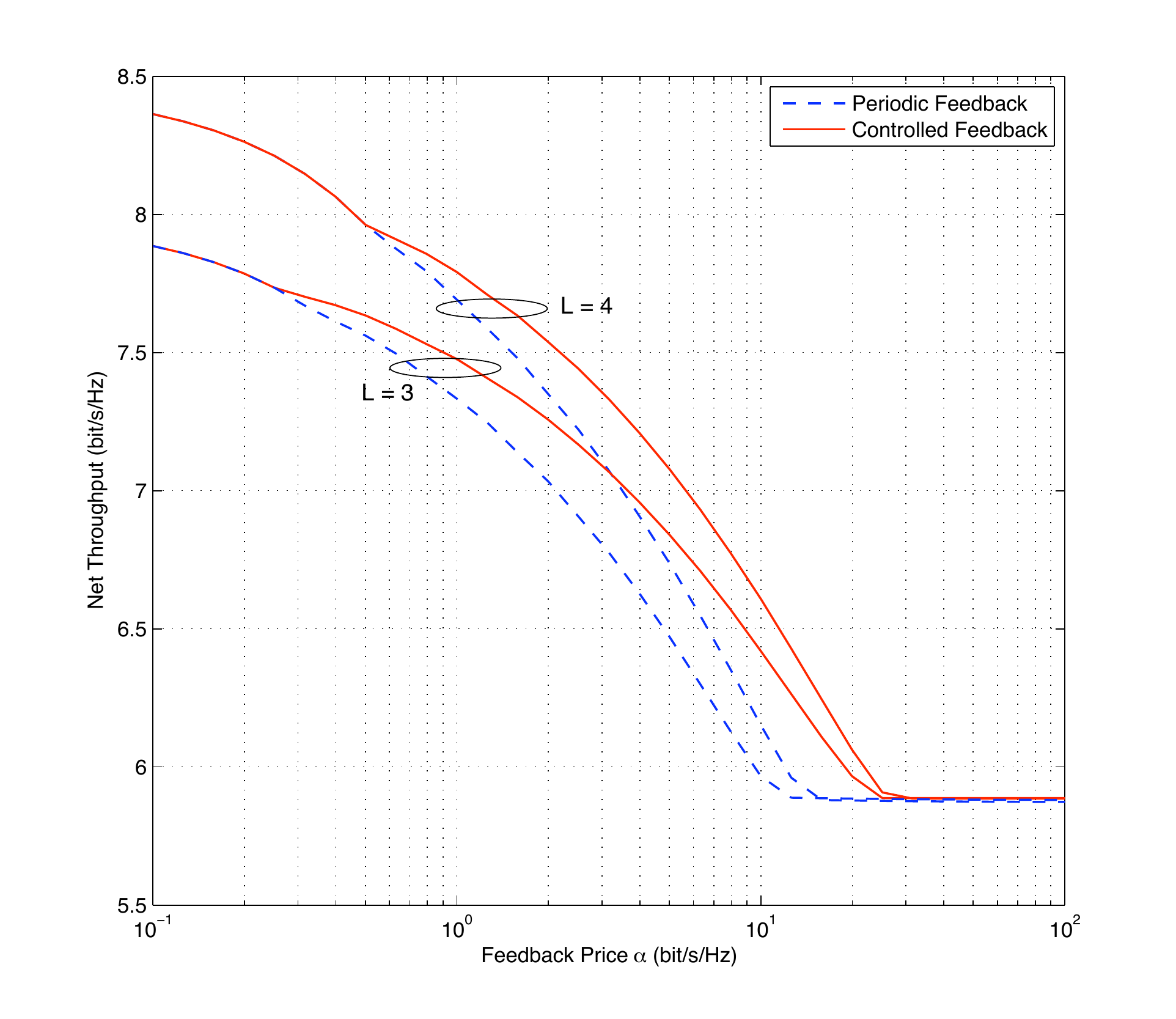}
\caption{Net throughput versus feedback price for both the optimally controlled and periodic feedback over the perfect feedback channel. The Doppler frequency is $f_D = 0.1/T_c$ and the numbers of transmit antenna $L=\{3, 4\}$.}
\label{Fig:ThputVarAnt}
\end{figure}

The comparison in Fig.~\ref{Fig:ThputVarDop} continues in Fig.~\ref{Fig:ThputVarAnt} but for different numbers of transmit antennas $L=\{3, 4\}$. It is observed that the maximum net throughput  gain for controlled feedback over the periodic feedback is about $0.5$ bit/s/Hz for both $L=3$ and $L=4$. Thus this gain is insensitive to the change on $L$.

Refer to  Footnote~\ref{Foot:MultiObj}.  The mentioned function of  maximum throughput versus  average feedback rate (normalized for $\alpha =1$) is plotted in Fig.~\ref{Fig:ThputVsFbCost} for $f_D = \{0.1, 0.01\}/T_c$ and $L=3$. Also plotted is the matching curve for periodic feedback obtained by a numerical search over different feedback intervals. As observed from the figure, for the same average feedback rate, optimal controlled  feedback provides  up to $0.5$ bit/s/Hz higher throughput than periodic feedback. Alternatively, given identical throughput, the former can reduce the feedback cost by half with respect to the latter (cf.  throughput $=7$ bit/s/Hz and $f_DT_c = 0.01$).

\begin{figure}
\centering\includegraphics[width=11cm]{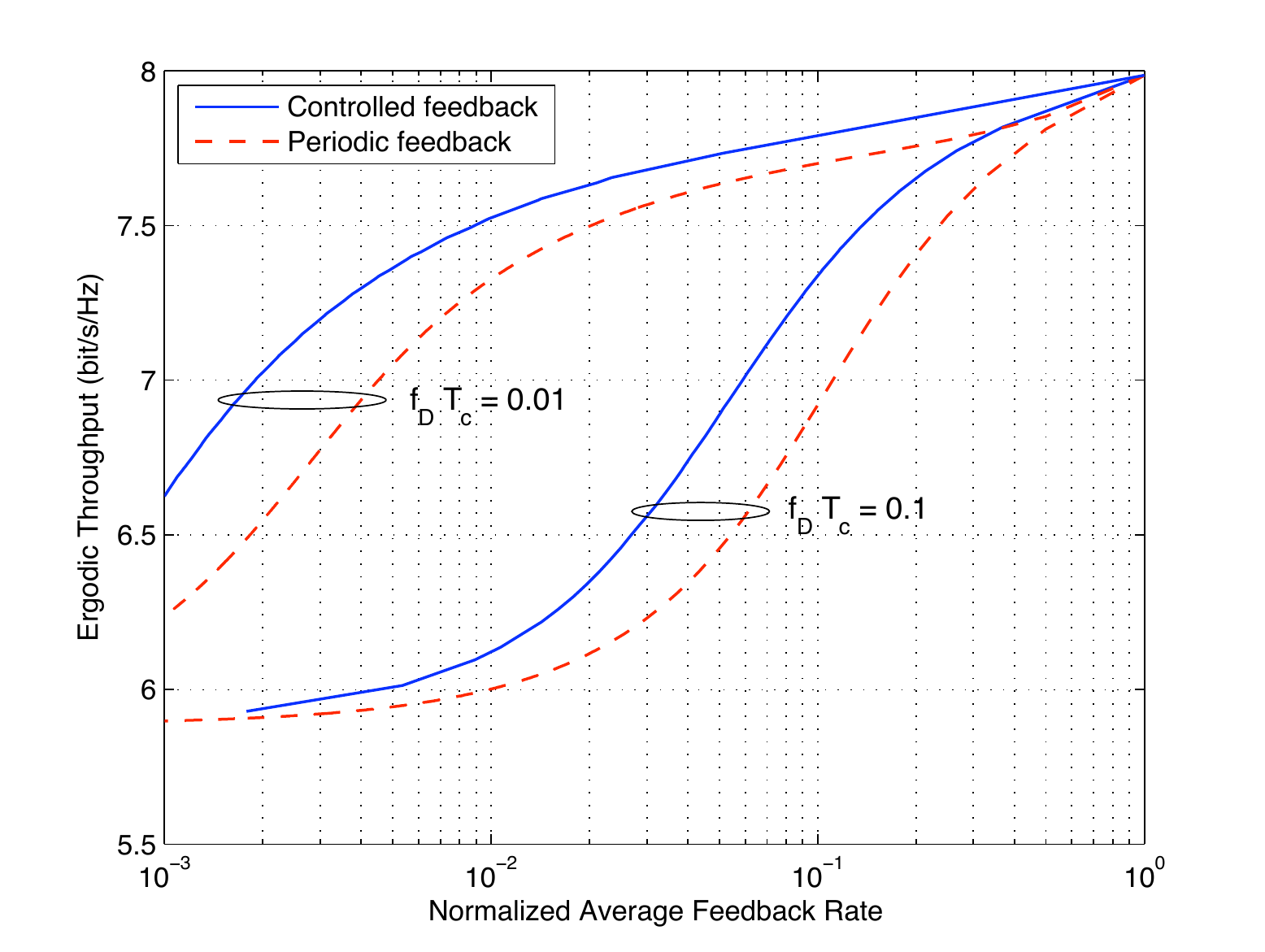}
\caption{Throughput  versus normalized average feedback rate for  the optimally controlled and periodic feedback over the perfect feedback channel. The Doppler frequency is $f_D = \{0.1, 0.01\}/T_c$ and the number of transmit antenna $L=3$. }
\label{Fig:ThputVsFbCost}
\end{figure}

Fig.~\ref{Fig:NetThput:ImpChan} displays the curves of net throughput versus feedback price for the \emph{perfect} and \emph{finite-rate} (quantized) feedback channels. The Doppler frequency is $f_D = 0.1/T_c$ and the number of transmit antennas $L=3$. The codebook used for quantizing feedback CSI has the size of $|\mathcal{F}|=16$ and is constructed using Lloyd's algorithm \cite{Lau:MIMOBlockFadingFbLinkCapConst:04,Xia:AchieveWelchBound:05}.  As observed from Fig.~\ref{Fig:ThputVarAnt}, feedback quantization reduces net throughput slightly. This loss is larger for smaller $\alpha$ (more frequent feedback) and vice versa.

\begin{figure}[t]
\centering\includegraphics[width=11cm]{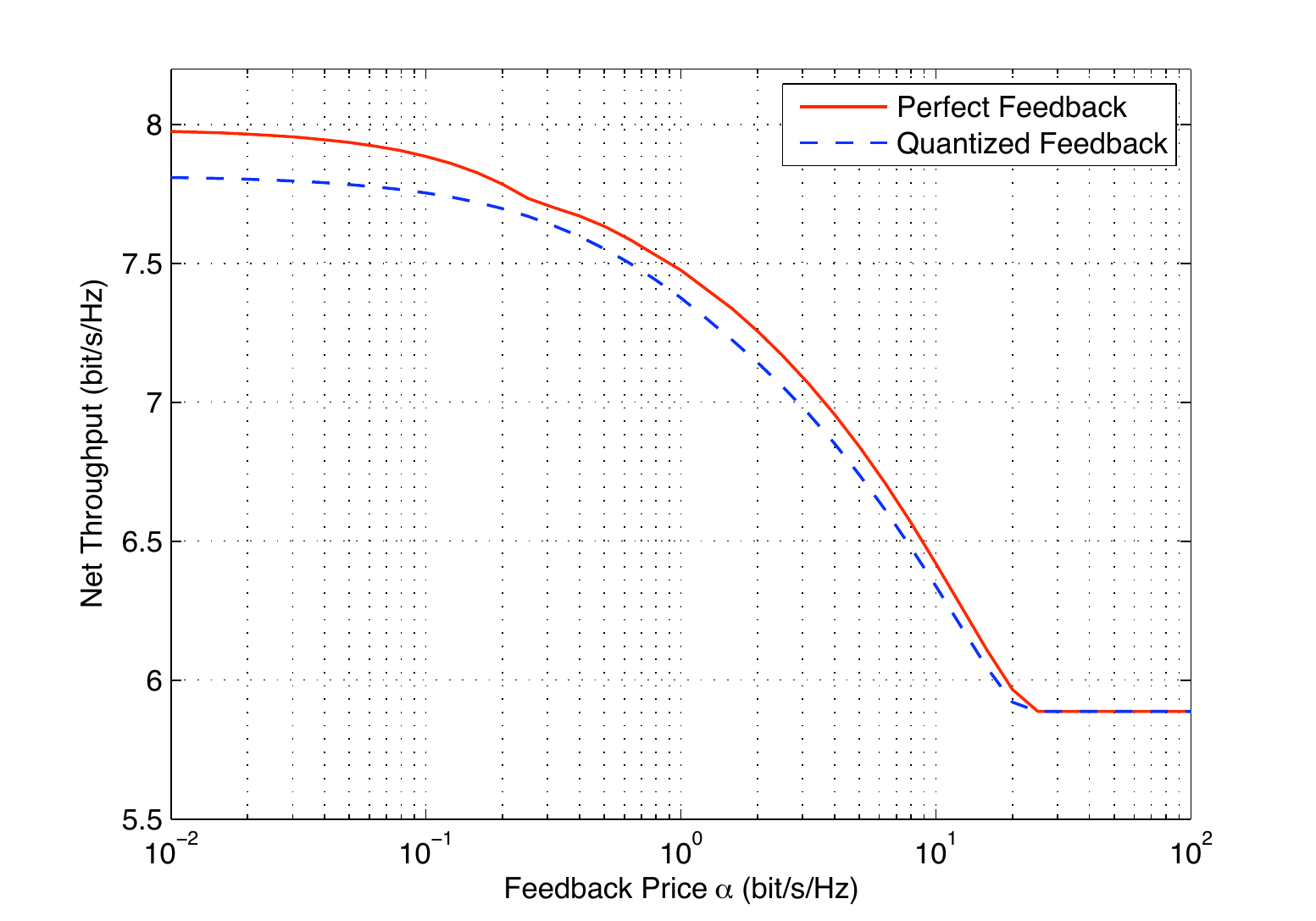}
\caption{Net throughput versus feedback price for  the optimally controlled feedback over the perfect and quantized feedback channels. The Doppler frequency is  $f_D = 0.1/T_c$; the number of transmit antennas $L=3$; the codebook used for quantizing feedback CSI has the size of $|\mathcal{F}|=16$.  }
\label{Fig:NetThput:ImpChan}
\end{figure}

\begin{figure}
\centering\includegraphics[width=11cm]{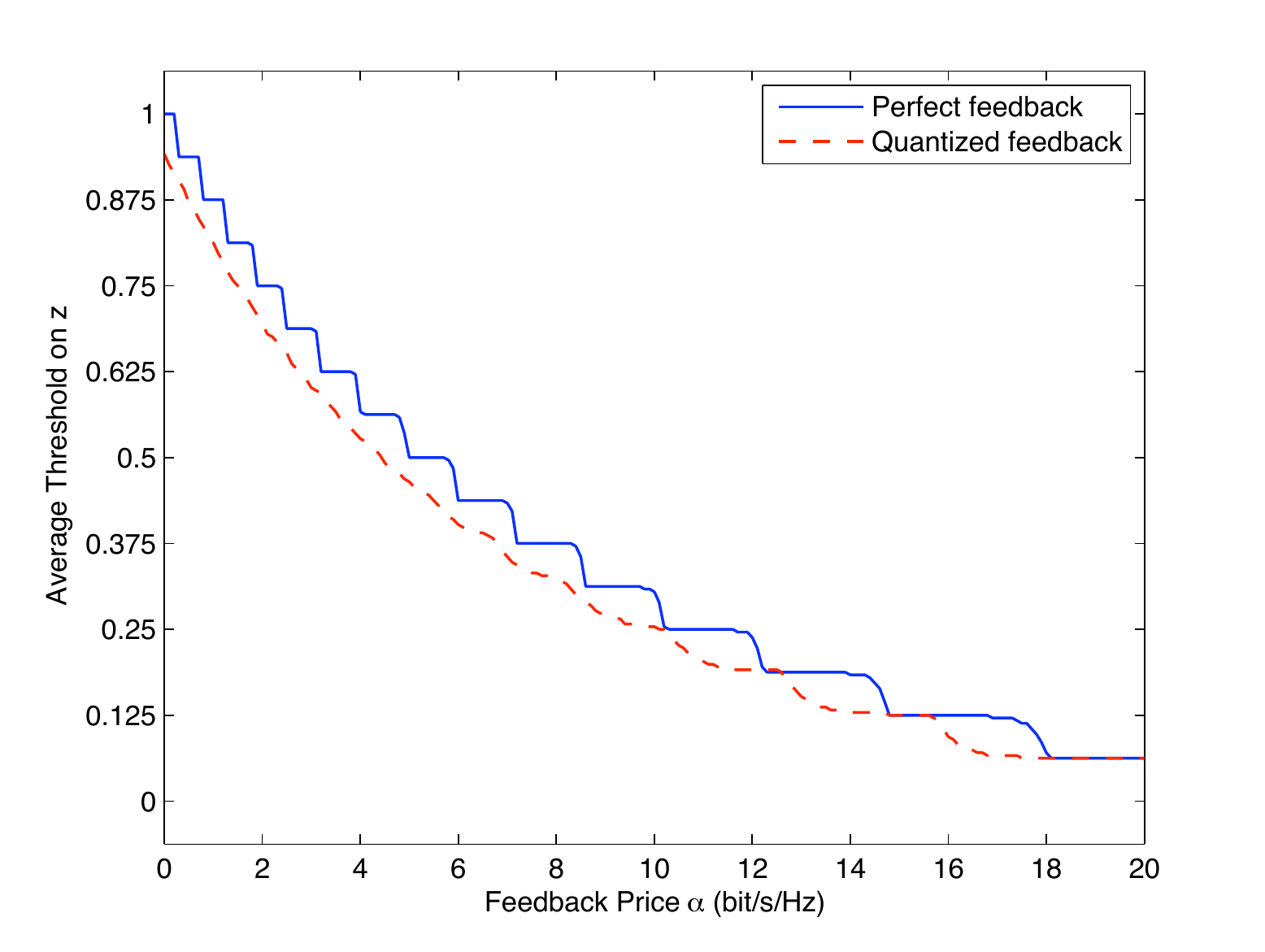}
\caption{Average threshold on $\hat{z}$ versus feedback price $\alpha$  for  the optimally controlled feedback over the perfect and quantized feedback channels. The Doppler frequency is  $f_D = 0.1/T_c$; the number of transmit antennas $L=3$; the codebook used for quantizing feedback CSI has the size of $|\mathcal{F}|=16$. }
\label{Fig:ThresholdVsPrice}
\end{figure}

The optimal control policies computed in simulation using policy iteration \cite{Bertseka:ConvergeDiscretDynamicProg:75} have the same threshold structure as predicted by analysis in the preceding sections. This also validates Assumption~\ref{AS:SD} on the channel  temporal correlation. The thresholds for these policies are observed to be insensitive to the variation on channel gain $\hat{g}$ (cf. Remark 4) on Theorem~\ref{Theo:Policy:Quant}). For this reason, the feedback threshold on $\hat{z}$ is averaged over the range of $\hat{g}$ and plotted against the feedback price $\alpha$ in Fig.~\ref{Fig:ThresholdVsPrice}, where both perfect and finite-rate feedback channels are considered.  The simulation parameters follow those for Fig.~\ref{Fig:ThputVarAnt}. As observed from Fig.~\ref{Fig:ThresholdVsPrice}, the average feedback threshold for the perfect feedback channel is $1$ for $\alpha = 0$, corresponding to feedback for every time slot; the threshold converges to zero as $\alpha$ increases. Fig.~\ref{Fig:ThresholdVsPrice} shows that feedback quantization reduces the feedback threshold slightly, implying less frequent feedback. Moreover, for $\alpha = 0$, the average feedback threshold for quantized feedback is smaller than one, agreeing with the remark on Corollary~\ref{Cor:Policy:QuantFb}.  Last, note that the humps on the curves in Fig.~\ref{Fig:ThresholdVsPrice} are caused by quantizing the controller's state space.

\section{Conclusion}\label{Section:Conclusion}
In this paper, we have proposed the approach of controlling feedback for maximizing net throughput of transmit beamforming systems. The optimal control policy has been proved to be of the threshold type. Under this policy, feedback is performed when the angle between transmit beamformer and the channel exceeds a threshold, which varies with the channel power. The threshold-type optimal policy has been shown to apply to both quantized and continuous controller inputs and both perfect and finite-rate feedback channels.  Feedback quantization has been found to decrease net throughput similarly as increasing the feedback price. As observed from simulation results, the optimal feedback control contributes significant net throughput gains without requiring additional bandwidth or  antennas. 

The work opens several issues for future investigation. First, the closed-form expression of the optimal feedback control policy can be derived by making additional assumptions on channel statistics. This allows direct policy computation rather than using the more complicated policy iteration method. Second, the  controlled feedback approach can be extended to other types of multi-antenna systems with feedback such as precoded spatial multiplexing or multiuser MIMO. Feedback in these systems supports multiple  operations such as spatial multiplexing, interference avoidance, and scheduling.  As a result,  the computation and analysis of optimal control policies are more challenging  than those for single-user transmit beamforming considered in this paper. Last, considering bursty data makes it necessary to jointly control the forward-link queue and CSI feedback. Addressing this issue by extending the approach in \cite{Berry:CommFadingChanDelayConst:2002} can establish an optimal tradeoff relation between feedback overhead, transmission power and queueing delay.

\appendix
\subsection{Proof of Lemma~\ref{Lem:MonoProd}}\label{App:MonoProd}
Let $K$ denote the height of $\bv$ and $\bA$. Define $\Delta v_\ell := [\bv]_\ell - [\bv]_{\ell-1}$ with $\bv_{-1} :=0$. Note that $\Delta v_\ell\geq 0$ for all $\ell$ since $\bv$ is monotone. Using the above definition, we can write 
\begin{equation}
[\bA^\dagger \bv]_n = \sum_{\ell=0}^{K-1} [\bv]_\ell [\bA]_{\ell, n} = \sum_{r=0}^{K-1}\sum_{\ell=r}^{K-1} \Delta v_r [\bA]_{\ell, n}.\label{Eq:proof:2}
\end{equation}
It follows that for $n_1\geq n_2$
\begin{eqnarray}
[\bA^\dagger \bv]_{n_1} -[\bA^\dagger \bv]_{n_2}
&=&  \sum_{r=0}^{K-1} \Delta v_r \l(\sum_{\ell=r}^{K-1}[\bA]_{\ell, n_1} - \sum_{\ell=r}^{K-1}[\bA]_{\ell, n_2}\r) \nn\\
&\overset{(a)}{\geq}& 0 \nn
\end{eqnarray}
where (a) follows from the monotonicity of $\bA$ and that $\Delta v_r\geq 0\ \forall \ r$.  The completes the proof of $1)$. 

For $n_1\geq n_2$, the difference between the $n_1$th and $n_2$th elements of the $k$th row of $\bB\bC$ is 
\begin{eqnarray}
\sum_{\ell}[\bB]_{k,\ell}[\bC]_{\ell, n_1} - \sum_{\ell}[\bB]_{k,\ell}[\bC]_{\ell, n_2} &=& \sum_{\ell}[\bB]_{k,\ell}([\bC]_{\ell, n_1} - [\bC]_{\ell, n_2})\nn\\
&\overset{(b)}{\geq}& 0
\end{eqnarray}
where $(b)$ follows from the monotonicity of each row of $\bC$ and that the elements of $\bB$ are nonnegative. Then $2)$ follows from the above inequality.

\subsection{Proof of Lemma~\ref{Lemma:Monotone}}\label{App:Monotone}
Consider fixed $m$, $n_1$ and  $n_2$ with $n_1\geq n_2$  and a nonnegative function $q(m, n)$ that increases monotonically with $n$. For instance, the all-zero function is a suitable choice. Based on the value iteration in \eqref{Eq:ValIt}, to prove the lemma,   it is sufficient to show that $(\mathsf{F} q)(m,n)$ is also a monotonically increasing function of $n$ given $m$.  Define $\mu'$ by 
\begin{equation}
(\mathsf{F} q)(m, n_2)=G(\bar{g}_m, \bar{z}_{n_2}, \mu') +\beta\sum_{k,\ell}q(k,\ell)\tilde{P}_{k,m}P_{\ell, n_2}(\mu'). \nn
\end{equation}
Define the matrix $\bQ$ with $[\bQ]_{k, \ell} = q(k,\ell)$. The above equation can be rewritten as 
\begin{equation}
(\mathsf{F} q)(m, n_2)=G(\bar{g}_m, \bar{z}_{n_2}, \mu') +\beta[\tilde{\bP}^\dagger \bQ  \bP(\mu')]_{m,n_2}. \label{Eq:App:DP}
\end{equation}
Assume $\mu'=1$. Then from \eqref{Eq:DPOperator} and \eqref{Eq:App:DP} 
\begin{eqnarray}
(\mathsf{F} q)(m, n_1) -(\mathsf{F} q)(m, n_2) &\geq&G(\bar{g}_m, \bar{z}_{n_1}, 1)+\beta[\tilde{\bP}^\dagger \bQ  \bP(1)]_{m,n_1}-G(\bar{g}_m, \bar{z}_{n_2}, 1) -\beta[\tilde{\bP}^\dagger \bQ  \bP(1)]_{m,n_2}\nn\\
&\overset{(a)}{=}&0\nn
\end{eqnarray}
where $(a)$ follows from that both $G(\bar{g}_m, \bar{z}_{n}, 1)$ and $P_{\ell, n}(1)$ are independent of $n$. 
Next, assume $\mu'=0$. It follows that
\begin{eqnarray}
(\mathsf{F} q)(m, n_1) -(\mathsf{F} q)(m, n_2) &\geq&G(\bar{g}_m, \bar{z}_{n_1}, 0)+\beta[\tilde{\bP}^\dagger \bQ  \bP(0)]_{m,n_1}-G(\bar{g}_m, \bar{z}_{n_2}, 0) -\beta[\tilde{\bP}^\dagger \bQ  \bP(0)]_{m,n_2}\nn\\
&\overset{(b)}{\geq}& \beta[\tilde{\bP}^\dagger \bQ  \bP(0)]_{m,n_1} -\beta[\tilde{\bP}^\dagger \bQ  \bP(0)]_{m,n_2}\nn\\
&\overset{(c)}{\geq}& 0 
\end{eqnarray}
where $(b)$ holds since $G(\bar{g}_m, \bar{z}_n, \mu')$ is a monotonically increasing function of $\bar{z}_n$. $(c)$ is due to that  the matrix $\tilde{\bP}^\dagger \bQ  \bP(0)$ has monotone rows, which results from Lemma~\ref{Lem:MonoProd} and that $\bP(0)$ is monotone and $\bQ$ comprises monotone rows. 
Combining above results proves the monotonicity of $\mathsf{F} q$ and hence $\hat{J}^\star_\beta$. 

Proving the monotonicity of $\sum_{k,\ell} \hat{J}^\star_\beta(k,\ell)\tilde{P}_{k,m}P_{\ell, n}$ uses that of  $\hat{J}^\star_\beta$ as shown above. The proof procedure is similar to  the above steps and thus omitted. 

\subsection{Proof of Theorem~\ref{Theo:Policy:Quant}}\label{App:Policy:Quant}
Consider the optimal policy $\hat{\mathcal{P}}_\beta^\star$ for maximizing the discounted reward. To simply notation, define 
\begin{equation}\begin{aligned}
&\Delta J(m,n) &:=& \log_2(1+P\hat{g}_m) -\log_2(1+P\hat{g}_m\hat{z}_n)- \alpha + \sum_{k,\ell}\hat{J}^\star_\beta(k,\ell) \tilde{P}_{k,m}P_{\ell, n}(1)-\\
&&&\sum_{k,\ell}\hat{J}^\star_\beta(k,\ell) \tilde{P}_{k,m}P_{\ell, n}(0).
\end{aligned}
\label{Eq:PF:DJ}
\end{equation}
From Bellman's equation $\hat{J}_\beta^\star = \mathsf{F}\hat{J}_\beta^\star$ with $\mathsf{F}$ defined in \eqref{Eq:DPOperator},  $\mathcal{P}^\star_\beta(\bar{g}_m,\bar{z}_n) = 1$ if $\Delta J(m, n)>0$ or otherwise $\mathcal{P}^\star_\beta(\bar{g}_m,\bar{z}_n) = 0$. Consider 
$(m_0, n_0)$ such that  $\Delta J(m_0,n_0) \leq 0$. For any $n$ with  $n_0 \leq n\leq N-1$, given that $P_{\ell, n}(1)$ is independent of $n$,  it follows from \eqref{Eq:PF:DJ} and Lemma~\ref{Lemma:Monotone} that $\Delta J(m_0, n)\leq 0$. Therefore for each $\hat{g}\in \hat{\mathcal{G}}$, there exists the matching  $\hat{\eta} \in \hat{\mathcal{Z}}$ such that $\hat{\mathcal{P}}^\star_\beta(\hat{g}, \hat{z})=0\ \forall \ \hat{z}\geq \hat{\eta}$ and $\hat{\mathcal{P}}^\star_\beta(\hat{g}, \hat{z})=1\ \forall \  \hat{z}<\hat{\eta}$ if $\hat{\eta} > \bar{z}_0$. Defining  $\hat{y}$ as the mapping from $\hat{g}$ to $\hat{\eta}$ proves that the optimal policy $\hat{\mathcal{P}}^\star_\beta$ is of the threshold type with $\hat{y}$ being the threshold function. 

Next, the bounds in \eqref{Eq:Threshold:Bounds} are proved as follows. The upper bound is trivial given that $0\leq \hat{z}\leq 1$. By the above definition,  $y(\hat{g})$ can be written as 
\begin{equation}\label{Eq:Threshold}
\hat{y}(\hat{g}) = \min_{\hat{z}\in\mathcal{Z}} \hat{z}\quad \textrm{s.t.}\quad \Delta J(\hat{g}, \hat{z}) \leq 0
\end{equation}
where by abuse of notation $\Delta J(\bar{g}_m, \bar{z}_n) := \Delta J(m, n)$. 
From\eqref{Eq:PF:DJ} and Lemma~\ref{Lemma:Monotone} 
\begin{equation}
\Delta J(m,n) \geq \underbrace{\log_2(1+P\bar{g}_m) -\log_2(1+P\bar{g}_m\bar{z}_n)- \alpha}_{\Delta J^-(m,n)}. \label{Eq:PF:DJ:a}
\end{equation}
It follows from \eqref{Eq:PF:DJ} and  Lemma~\ref{Lemma:Monotone} that $\Delta J(m,n)$ is a monotonically increasing function of $n$. So is $\Delta J^-(m,n)$ from its definition. Therefore, from \eqref{Eq:Threshold} and \eqref{Eq:PF:DJ:a}
\begin{equation}\label{Eq:Threshold:Ineq}
\hat{y}(\hat{g}) \geq \hat{\gamma}
\end{equation}
where 
\begin{equation}\label{Eq:Threshold:LB:a}
\hat{\gamma} = \min_{\hat{z}\in\mathcal{Z}} \hat{z}\quad \textrm{s.t.}\quad \Delta J^-(d(\hat{g}), d(\hat{z})) \leq 0.
\end{equation}
Using \eqref{Eq:Threshold:Ineq} and solving for $\hat{\gamma}$ using \eqref{Eq:Threshold:LB:a} proves that the lower bound in \eqref{Eq:Threshold:Bounds} holds for $\hat{\mathcal{P}}_\beta^\star$.

Given $0<\beta_0<1$, the properties for  $\hat{\mathcal{P}}^\star_\beta$ as proved above hold   for any $\beta\in [\beta_0, 1)$. These properties must also exist for  the optimal policy giving   $\hat{J}^\star = \lim_{\beta\rightarrow 1}(1-\beta)\hat{J}^\star_\beta$ \cite{Bertsekas07:DynamicProg}. This completes the proof. 

\subsection{Proof of Corollary~\ref{Cor:Policy:ExtremeFb}}\label{App:Policy:ExtremeFb}
The first claim is obviously valid since for $\alpha = \infty$ any feedback instant causes net throughput to be $-\infty$ and thus the optimal feedback controller should block feedback by using the threshold $\hat{y} = 0$. The second claim holds since for $\alpha = 0$, $G(\hat{g}, \hat{z}, 0) \leq G(\hat{g}, \hat{z}, 1) \ \forall \ (\hat{g}, \hat{z})\in \hat{X}$. Thus feedback should be performed in every time slot, corresponding to fixed $\hat{y} = 1$. 

\subsection{Proof of Proposition~\ref{Prop:Policy:Cont}}\label{App:Policy:Cont}
A stationary  feedback policy partitions the continuous state space $\mathcal{X}$ into two sets $\mathcal{W}$ and $\mathcal{W}^c$ such that $\mu=1 \  \forall \ (g,z)\in\mathcal{W}$ and $\mu=0 \  \forall \ (g,z)\in\mathcal{W}^c$. To simplify notation, define $\mathcal{W}(g) := \mathcal{W}\cap (\{g\}\times\mathcal{Z})$ and $ \mathcal{W}^c = \mathcal{W}\backslash \mathcal{W}(g)$. Moreover, let $f(z\mid \mathcal{W})$ denote the PDF of $z$ that depends on the set (policy) $\mathcal{W}$. Note that the PDF $f(g)$ of $g$ is independent of $\mathcal{W}$. 

To apply Lemma~\ref{Lem:Converge}, we design a genie-aided dummy feedback-control system similar to the current one but with a bounded continuous state space. In the virtual system, the encoder is   
shut down by the genie  whenever $g> \tilde{g}_{M-1}$ or otherwise turned on. The average reward for this system is $I := J(G')$
where the reward-per-stage $G'$ is defined in terms of $G$ in \eqref{Eq:RewardPerStage} as
\begin{equation}
G'(g, z, \mu) = \l\{
\begin{aligned}
&G(g,z, \mu),&& g\geq \tilde{g}_{M-1}\\
&0,&& \textrm{otherwise}. 
\end{aligned}
\r.
\end{equation}
The maximum reward $I^\star$ can be written as 
\begin{equation}\label{Eq:VirRew}\begin{aligned}
\!\!\!\!I^\star(M) &=&\!\!\!\!\!\!   \max_{\mathcal{W}\in\mathcal{G}\times\mathcal{Z}} \int\limits^{\tilde{g}_{M-1}}_0 \l\{\int\limits_{\mathcal{W}(g)}G'(g, z, 1) f(z\mid \mathcal{W})dz +\!\!\!\!
 \int\limits_{\mathcal{W}^c(g)}G'(g, z, 0) f(z\mid \mathcal{W})dz\r\}f(g)dg.
\end{aligned}
\end{equation}
Next, the reward $I^\star(M)$ is shown to converge to $J^\star$ as $N$ increases. Similar to \eqref{Eq:VirRew},
\begin{eqnarray}
J^\star 
&=& \max_{\mathcal{W}\in\mathcal{G}\times\mathcal{Z}} \int\limits^{\infty}_0 \l\{\int\limits_{\mathcal{W}(g)}G'(g, z, 1) f(z\mid \mathcal{W})dz +
 \int\limits_{\mathcal{W}^c(g)}G'(g, z, 0) f(z\mid \mathcal{W})dz\r\}f(g)dg\nn\\
&\leq&  \max_{\mathcal{W}\in\mathcal{G}\times\mathcal{Z}} \int\limits^{_{\tilde{g}_{M-1}}}_0 \l\{\int\limits_{\mathcal{W}(g)}G'(g, z, 1) f(z\mid \mathcal{W})dz +
 \int\limits_{\mathcal{W}^c(g)}G'(g, z, 0) f(z\mid \mathcal{W})dz\r\}f(g)dg+\nn\\
&&\max_{\mathcal{W}\in\mathcal{G}\times\mathcal{Z}} \int\limits^{\infty}_{\tilde{g}_{M-1}} \l\{\int\limits_{\mathcal{W}(g)}G'(g, z, 1) f(z\mid \mathcal{W})dz +
 \int\limits_{\mathcal{W}^c(g)}G'(g, z, 0) f(z\mid \mathcal{W})dz\r\}f(g)dg\nn\\
&\leq& I^\star(M)+\max_{\mathcal{W}\in\mathcal{G}\times\mathcal{Z}}\!\!\!\! \int\limits^{\infty}_{\tilde{g}_{M-1}}\!\! \l\{\log_2(1+Pg)\int\limits_{\mathcal{W}(g)}\!\!\!\! f(z\mid \mathcal{W})dz +
\log_2(1+Pg) \int\limits_{\mathcal{W}^c(g)} \!\!\!\! f(z\mid \mathcal{W})dz\r\}f(g)dg\nn\\
&=& I^\star(M)+\int\limits^{\infty}_{\tilde{g}_{M-1}} \log_2(1+Pg)f(g)dg\nn\\
&\overset{(a)}{\leq} & I^\star(M)+\frac{\E[ \log_2(1+Pg)]}{\tilde{g}_{M-1}}\label{Eq:IJ}
\end{eqnarray}
where $(a)$ is obtained by applying Markov's inequality. Given that $g$ follows chi-square distribution and $\Pr(g\geq  \tilde{g}_{M-1}) = 1/M$, $\tilde{g}_{M-1}\rightarrow\infty$ for $M\rightarrow\infty$. Therefore, since  $I^\star$ and $\E[\log_2(1+Pg)]$ are finite, it follows form $\eqref{Eq:IJ}$ that
\begin{equation}\label{Eq:IJ:Converge}
\lim_{M\rightarrow\infty}|J^\star - I^\star(M)| = 0. 
\end{equation}

Next, consider the approximate  feedback control optimization for the virtual system with a discrete state space. This space, denoted as  $\hat{\mathcal{S}}$, is obtained using a quantization algorithm similar to that in Section~\ref{Section:QuantAlgo}, hence  $\hat{\mathcal{S}} = (\hat{\mathcal{G}}\backslash\{\bar{g}_{M-1}\})\times \mathcal{Z}$. Let $\hat{I}^\star_\beta$ and $\hat{I}^\star$ denote the maximum discounted and average rewards, respectively. For the above approximated problem, the maximum quantization error is given as 
\begin{equation}
d_s = \max_{0\leq m\leq M-2}\max_{0\leq n\leq N-1}\max_{\tilde{g}_m\leq g\leq \tilde{g}_{m+1}}\max_{\tilde{z}_n\leq z\leq \tilde{z}_{n+1}} \sqrt{|g-\bar{g}_m|^2 + |z-\bar{z}_n|^2}.\label{Eq:MaxErr}
\end{equation}
Since  $d_s\rightarrow 0$ as $M,N\rightarrow \infty$ and  using the Lipschitz continuity of the conditional PDF's of $(g, z)$ and the reward-per-stage function, it follows from Lemma~\ref{Lem:Converge} that 
\begin{equation}\label{Eq:II:Converge}
\lim_{M,N\rightarrow\infty}|\hat{I}^\star(M,N) - I^\star| = \lim_{\beta\rightarrow 1}\lim_{N\rightarrow\infty} (1-\beta)|\hat{I}^\star_\beta(\hat{g}, \hat{z}, M, N) - I^\star_\beta(g, z)| = 0
\end{equation}
From \eqref{Eq:IJ:Converge} and \eqref{Eq:II:Converge} and the triangular inequality 
\begin{equation}\label{Eq:IJ:Converge:a}
\lim_{M,N\rightarrow\infty}|J^\star - \hat{I}^\star(N)| \leq \lim_{M,N\rightarrow\infty}\l(|J^\star - I^\star(M)| + |I^\star - \hat{I}^\star(M,N)|\r) = 0. 
\end{equation}
Furthermore, for $M,N\rightarrow\infty$, the results in Theorem~\ref{Theo:Policy:Quant} holds for the virtual system with the state space $\hat{\mathcal{S}}$. This completes the proof.

\subsection{Proof of Proposition~\ref{Prop:AvAward:LimFb}}\label{App:AvAward:LimFb}

\emph{Proof of the inequalities in $1)$ and $2)$:} The second inequality in $1)$ holds since  $G \geq G_\epsilon$ from their definitions in \eqref{Eq:RewardPerStage} and \eqref{Eq:RewardPerStage:Quant}. 
In the sequel, we prove the first inequality in $1)$ based on value iteration \cite{Bertsekas07:DynamicProg}. To this end, consider two nonnegative functions $q_1(\hat{g}, \hat{z})$ and $q_2(\hat{g},\hat{z})$ that have the support $\hat{\mathcal{G}}\times \hat{\mathcal{Z}}$ and monotonically increase with $\hat{z}$. Furthermore, $q_1(\hat{g}, \hat{z})\geq q_2(\hat{g},\hat{z})\ \forall \ (\hat{g}, \hat{z})\in \hat{\mathcal{G}}\times \hat{\mathcal{Z}}$, which is  represented by $q_1 \geq q_2$ for simplicity. Following the similar procedure as in the proof of Lemma~\ref{Lemma:Monotone}, it can be  shown that the functions $\mathsf{F}(G, \bP)q_1$ and  $\mathsf{F}(G, \bP_{\epsilon})q_2$ both monotonically increases with $\hat{z}$, where $\mathsf{F}$ is in \eqref{Eq:DPOperator}.

Next, it is shown that $\mathsf{F}(G, \bP)q_1 \geq \mathsf{F}(G, \bP_{\epsilon})q_2$. Let $\mu_a$ and $\mu_b$ denote the control decisions that satisfy 
\begin{eqnarray}
[\mathsf{F}(G, \bP)q_1](\bar{g}_m, \bar{z}_n) &=& G(\bar{g}_m, \bar{z}_n, \mu_a) + \sum_{k,\ell}q_1(\bar{g}_k,\bar{z}_{\ell})[\tilde{\bP}]_{k,m} [\bP(\mu_a)]_{\ell, n}\\
\l[\mathsf{F}(G, \bP_{\epsilon})q_2\r](\bar{g}_m, \bar{z}_n) &=& G(\bar{g}_m, \bar{z}_n, \mu_b) + \sum_{k,\ell}q_2(\bar{g}_k,\bar{z}_{\ell})[\tilde{\bP}]_{k,m} [\bP_{\epsilon}(\mu_b)]_{\ell, n}
\end{eqnarray}
If $\mu_a = \mu_b = 1$,  
\begin{eqnarray}
[\mathsf{F}(G, \bP)q_1-\mathsf{F}(G, \bP_{\epsilon})q_2](\bar{g}_m, \bar{z}_n) &=& \log_2(1+P\bar{g}_m) + \sum_{k,\ell}q_1(\bar{g}_k,\bar{z}_{\ell})[\tilde{\bP}]_{k,m} [\bP(1)]_{\ell, n}\nn\\
&& - \underset{\epsilon}{\E}[\log_2(1+P\bar{g}_m\epsilon)] - \sum_{k,\ell}q_2(\bar{g}_k,\bar{z}_{\ell})[\tilde{\bP}]_{k,m} [\bP_\epsilon(1)]_{\ell, n} \nn\\
&\geq& \sum_{k,\ell}q_1(\bar{g}_k,\bar{z}_{\ell})[\tilde{\bP}]_{k,m} [\bP(1)]_{\ell, n} - \sum_{k,\ell}q_2(\bar{g}_k,\bar{z}_{\ell})[\tilde{\bP}]_{k,m} [\bP_\epsilon(1)]_{\ell, n} \nn\\
&\overset{(a)}{\geq}& \sum_{k,\ell}[\tilde{\bP}]_{k,m}q_2(\bar{g}_k,\bar{z}_{\ell}) \l\{[\bP(1)]_{\ell, n} - [\bP_\epsilon(1)]_{\ell, n}\r\}. \label{Eq:OpDiff}
\end{eqnarray}
where $(a)$ follows from $q_1 \geq q_2$. 
For $0\leq \ell_0\leq N-1$ and from \eqref{Eq:TxProb:z} and \eqref{Eq:TransProb:Eps}
\begin{eqnarray}
\sum_{\ell = \ell_0}^{N-1}[\bP(1)]_{\ell, n} - \sum_{\ell = \ell_0}^{N-1}[\bP_\epsilon(1)]_{\ell, n} &=& \int_{\hat{z}=\tilde{z}_{\ell_0}}^1  \check{f}(\hat{z}\mid \hat{z}'=1) d\hat{z} - \int_{\hat{z}=\tilde{z}_{\ell_0}}^1  \underset{\epsilon}{\E}[\check{f}(\hat{z}\mid \hat{z}'=\epsilon)] d\hat{z}\nn
\end{eqnarray}
\begin{eqnarray}
&=& \underset{\epsilon}{\E}\l[\int_{\hat{z}=\tilde{z}_{\ell_0}}^1  \check{f}(\hat{z}\mid \hat{z}'=1) d\hat{z} - \int_{\hat{z}=\tilde{z}_{\ell_0}}^1  \check{f}(\hat{z}\mid \hat{z}'=\epsilon) d\hat{z}\r]\nn\\
&\overset{(b)}{\geq}& 0 \label{Eq:OpDiff:a}. 
\end{eqnarray}
where $(b)$ is due to Assumption~\ref{AS:SD}. 
Using \eqref{Eq:OpDiff} and \eqref{Eq:OpDiff:a} and following the similar steps as in the proof for Lemma~\ref{Lemma:Monotone} leads to that $\mathsf{F}(G, \bP)q_1 \geq \mathsf{F}(G, \bP_\epsilon)q_2$ if $\mu_a = \mu_b = 1$. If $\mu = \mu = 0$, since $\bP_\epsilon(0) = \bP(0)$ and $q_1 \geq q_2$, 
\begin{eqnarray}
[\mathsf{F}(G, \bP)q_1-\mathsf{F}(G, \bP_{\epsilon})q_2](\bar{g}_m, \bar{z}_n) &=& \sum_{k,\ell}[q_1(\bar{g}_k,\bar{z}_{\ell})-q_2(\bar{g}_k,\bar{z}_{\ell})][\tilde{\bP}]_{k,m} [\bP(0)]_{\ell, n}\nn\\
&\geq& 0. 
\end{eqnarray}
From \eqref{Eq:DPOperator},  the values of $[\mathsf{F}(G, \bP)q_1 - \mathsf{F}(G, \bP_{\epsilon})q_2]$ for  $(\mu = 1, \mu = 0)$ and $(\mu = 0, \mu = 1)$ are larger than those for  $(\mu = \mu = 0)$ and $(\mu = \mu = 1)$, respectively. Combining the above results shows that $\mathsf{F}(G, \bP)q_1\geq =\mathsf{F}(G, \bP_{\epsilon})q_2$. 

Consequently,  $\hat{J}^\star_\beta(G, \bP) \geq \hat{J}^\star_\beta(G, \bP_\epsilon)$ since by value iteration 
\begin{equation}
\hat{J}^\star_\beta (G, \bP)=\lim_{n\rightarrow\infty}\mathsf{F}^n(G, \bP)q_1 \quad\textrm{and}\quad  \hat{J}^\star_\beta(G, \bP_\epsilon)=\lim_{n\rightarrow\infty}\mathsf{F}^n(G, \bP_\epsilon)q_2. 
\end{equation}
As  $\hat{J}^\star = \lim_{\beta\rightarrow 1}(1-\beta)\hat{J}^\star_\beta$, the first inequality in 1) of the proposition statement is proved. 

The inequalities in $2)$ of the proposition statement can be proved also using the above procedure.

\emph{Proof of the inequalities in $3)$ and $4)$:} The inequalities in $3)$ and $4)$ can be proved using similar procedures. Thus we focus on proving that in $3)$. The reward-per-stage function in \eqref{Eq:RewardPerStage:Quant} can lower bounded below, where the fourth argument is the weighted feedback price
\begin{eqnarray}
G_\epsilon(\hat{g}, \hat{z}, 1, \alpha) &=&\underset{\epsilon}{ \E}\l[\log_2\l(1/\epsilon + P\hat{g}\r)\r] +\underset{\epsilon}{\E}[\log_2\epsilon]-\alpha \nn\\
&\overset{(a)}{\geq}& \log_2\l(1 + P\hat{g}\r) + \underset{\epsilon}{\E}[\log_2\epsilon] -\alpha\nn\\
&=& G(\hat{g}, \hat{z}, 1, \alpha - \underset{\epsilon}{\E}[\log_2\epsilon]) \label{Eq:Reward:LB}
\end{eqnarray}
where $(a)$ uses $\delta \leq 1$. Combining \eqref{Eq:Reward:LB} and $G_\epsilon(\hat{g}, \hat{z}, 0) = G(\hat{g}, \hat{z}, 0)$ gives the desired result.

\renewcommand{\baselinestretch}{1.3}
\bibliographystyle{ieeetr}

\end{document}